%% file: Rainbow.tex
\documentclass[journal]{IEEEtran}
\usepackage{amsmath,amssymb}
\usepackage{graphicx,graphics,color,psfrag}
\usepackage{cite,balance}
\usepackage{caption}
\usepackage{framed}
\allowdisplaybreaks
\usepackage{algorithm}
\usepackage{algorithmic}
\usepackage{accents}
\usepackage{amsthm}
\usepackage{bm}
\usepackage{url}
\usepackage[english]{babel}
\usepackage{multirow}
\usepackage{enumerate}
\usepackage{cases}
\usepackage{stfloats}
\usepackage{dsfont}
\usepackage{color,soul}
\usepackage{amsfonts}
\usepackage{cite,graphicx,amsmath,amssymb}
\usepackage{fancyhdr}
\usepackage{hhline}
\usepackage{graphicx,graphics}
\usepackage{array,color}
\usepackage{mathtools}
\usepackage{amsmath}
\usepackage[T1]{fontenc}
\usepackage{float}  %设置图片浮动位置的宏包
\usepackage{stfloats}
\usepackage{siunitx}
\usepackage{subfloat}

\include{header}

\usepackage[caption=false,font=footnotesize,labelfont=rm,textfont=rm]{subfig}

\def\BibTeX{{\rm B\kern-.05em{\sc i\kern-.025em b}\kern-.08em
		T\kern-.1667em\lower.7ex\hbox{E}\kern-.125emX}}
\usepackage{balance}

\begin{document}
\captionsetup[figure]{name={Fig.}}
\title{\huge Super-resolution Wideband Beam Training for Near-field Communications with Ultra-low Overhead} 
\author{Cong Zhou, Changsheng You,~\IEEEmembership{Member,~IEEE}, Shuo Shi,~\IEEEmembership{Member,~IEEE}, Jiasi Zhou, Chenyu Wu
	% <-this % stops a space
	\thanks{Cong Zhou, Shuo Shi and Chenyu Wu are with the School of Electronics and Information Engineering, Harbin Institute of Technology, Harbin, 150001, China. (e-mail: zhoucong@stu.hit.edu.cn, crcss@hit.edu.cn, wuchenyu@hit.edu.cn).}
	\thanks{Changsheng You is with State Key Laboratory of Optical Fiber and Cable Manufacture Technology, and also with the Department of Electronic and Electrical Engineering, Southern University of Science and Technology, Shenzhen 518055, China. (e-mails: youcs@sustech.edu.cn).}
	\thanks{Jiasi Zhou is with the School of Medical Information and Engineering, Xuzhou Medical University, Xuzhou, 221004, China. (e-mails: jiasi{\_}zhou@xzhmu.edu.cn).}% <-this % stops a space
	%\thanks{Manuscript received April 19, 2021; revised August 16, 2021.} 
	\thanks{\emph{Corresponding author: Shuo Shi and Changsheng You.}}
}

\maketitle
\begin{abstract}
	In this paper, we propose a super-resolution wideband beam training method for \emph{near-field} communications, which is able to achieve ultra-low overhead.
	To this end, we first study the multi-beam characteristic of a sparse uniform linear array (S-ULA) in the wideband.
	Interestingly, we show that this leads to a new beam pattern property, called \emph{rainbow blocks}, where the S-ULA generates multiple grating lobes and each grating lobe is further splitted into multiple versions in the wideband due to the well-known \emph{beam-split} effect.
	As such, one directional beamformer based on S-ULA is capable of generating multiple rainbow blocks in the wideband, hence significantly extending the beam coverage. 
	Then, by exploiting the beam-split effect in both the frequency and spatial domains, we propose a new \emph{three-stage} wideband beam training method for extremely large-scale array (XL-array) systems.
	Specifically, we first sparsely activate a set of antennas at the central of the XL-array and judiciously design the time-delay (TD) parameters to estimate candidate user angles by comparing the received signal powers at the user over subcarriers.
	Next, to resolve the angular ambiguity introduced by the S-ULA, we activate all antennas in the central subarray and design an efficient subcarrier selection scheme to estimate the true user angle.
	In the third stage, we resolve the user range at the estimated user angle with high resolution by controlling the splitted beams over subcarriers to simultaneously cover the range domain.
	Finally, numerical results are provided to demonstrate the effectiveness of proposed wideband beam training scheme, which only needs three pilots in near-field beam training, while achieving near-optimal rate performance.
\end{abstract}

%----------------------------------------------------------------------------------第一部分
\section{Introduction}
Next-generation wireless systems are migrating to higher frequency bands, such as millimeter-wave (mmWave) and terahertz (THz), for accommodating data-demanding applications like virtual/augmented reality~\cite{han2019terahertz,tan2022thz, Chen2019survey}.
Moreover, since more and more antennas can be packed in a small area thanks to shorter wavelengths of high-frequency bands, \emph{extremely large-scale arrays} (XL-arrays) are envisioned to be feasible in future wireless communication, which significantly increase the spectrum efficiency and spatial resolution \cite{nf_mag,cong2024near,nf_bm, zhang2023near}.

However, high-frequency bands and XL-arrays also incur two challenges. 
First, the XL-array fundamentally changes the radio propagation modeling, shifting from the far-field planar wavefronts to the near-field \emph{spherical wavefronts} (SWs), which brings both blessing and misfortune \cite{an2024near, lu2023tutorial}.
On one hand, the \emph{beam focusing} phenomenon greatly mitigates the inter-user interference (IUI) and improves the charging efficiency for wireless power transfer (WPT)~\cite{zhang20236g, xu2024resource, wpt}. 
In addition, the beam focusing property in near-field communications also enhances the so-called location division multiple access (LDMA) proposed in \cite{ldma} (instead of space division multiple access (SDMA) in  far-field communications), which can simultaneously serve different users at the same angle and allow for massive connectivity.
Moreover, the near-field multiple-input multiple-output (MIMO) channel demonstrates a higher channel rank and achieves a high multiplexing gain, even in the line-of-sight (LoS) scenario \cite{liu_review}.
Besides wireless communications, near-field SWs also enhance radar sensing, enabling the estimation of both target angle and range using conventional subspace-based algorithms (e.g., the multiple signal classification (MUSIC) method).
This eliminates the need for distributed arrays with system synchronization, which are typically required in the far-field for target range estimation \cite{cong2024near}.
On the other hand, for beam training, near-field communications need codewords sampled in both the angle and range domains, which leads to severe beam training overhead~\cite{you2024next}.

Second, since the phase-shifter (PS) based beamforming structure is frequency-independent, the extremely large bandwidth in high-frequency bands will introduce the so-called beam-split effect \cite{park2022beam}, for which the beams generated at different subcarriers (or frequencies) will be focused on varying locations and deviate from the desired location, resulting in significant performance degradation \cite{cui2024near, su2023wideband}. 
In this paper, we propose a new near-field wideband beam training method with super-resolution and ultra-low overhead by judiciously designing the time-delay (TD) beamforming to control the beam-split effect.
\vspace{-5pt}
\subsection{Prior Works}
To fully exploit the spatial multiplexing gain of XL-arrays, channel state information (CSI) acquisition is essential for beamforming designs.
Among others, channel estimation is an efficient approach to acquire explicit CSI by using e.g., least square (LS), minimum mean-squared error (MMSE) and compressed sensing (CS) methods. In particular, for low-complexity channel estimation, conventional compressed sensing (CS) techniques based on the discrete Fourier transform (DFT) codebook may face challenges in near-field scenarios, due to the characteristics of SWs.
Specifically, the near-field channel is no longer sparse in the angular domain, while it exhibits sparsity in the polar domain \cite{Cui2022channel}.
Based on this observation, the authors in \cite{Cui2022channel} proposed a polar-domain codebook, where the angle and range domains are uniformly and non-uniformly sampled, to sparsely represent the near-field channel.
As such, orthogonal matching pursuit (OMP) techniques can be utilized to recover the sparse channel. 
In addition, several parametric channel estimation methods have been proposed for near-field communications to enhance estimation accuracy.
For example, the authors in~\cite{huang2023low} proposed a low-complexity localization method, which decoupled the angle and range estimation for each path by leveraging the anti-diagonal elements of the covariance matrix.
As such, the LS criterion can be applied to estimate the channel gain for each path and then the near-field multi-path channel can be directly recovered.
Moreover, some studies exploited the spatial non-stationarity of the near-field channel to further reduce the required number of pilots.
For example, the authors in \cite{chen2023non} proposed to use a group time block code to transform the near-field spatial non-stationary channel into a series of stationary channels, followed by the application of the simultaneous OMP (SOMP) algorithm to estimate the non-stationary channel.

Alternately, beam training is another efficient method to inexplicitly acquire the CSI, by performing efficient beam sweeping over a predefined beam codebook and selecting the best beam codeword that yields the maximum received power at the user.
Note that beam training is particularly efficient in high-frequency bands due to the channel sparsity and robustness to noise \cite{liu_review}. Specifically, in the low-SNR regime, the accuracy of typical channel estimation methods (e.g., LS) may not work well, while beam training methods can provide more accurate CSI since it achieves a high received signal power when the beam aligns well with the user channel \cite{zheng2022survey,zhou2024multi}.
For practical implementation, various beam training methods for far-field communications (e.g., hierarchical search schemes) have been documented in standardization including IEEE 802.15.3c and IEEE 802.11ad \cite{xiao2016hierarchical}.
For near-field beam training, the authors in \cite{Cui2022channel} utilized the polar-domain codebook to determine the optimal near-field codeword in an exhaustive-search manner. 
Although this method can effectively resolve the energy-spread issue when conventional DFT codebook is applied in the near-field communications, it suffers from unaffordable beam training overhead, which is proportional to the product of the numbers of antennas and sampled ranges. 
To solve this problem, efficient near-field beam training methods with low overhead have been proposed which can be largely classified into two categories, namely, narrow-band and wideband beam training methods.
\subsubsection{Narrow-band Beam Training}
Compared to wideband-based methods, narrow-band beam training methods leverage only a single subcarrier per time slot, conserving spectral resources but at the cost of higher temporal resource consumption.
Specifically, the authors in \cite{two_phase} proposed a two-phase near-field beam training scheme, which decoupled the angle and range estimation.
They used the energy-spread effect in the received beam pattern with the DFT codebook, for which the user angle was estimated by the middle of a predefined angular support region and the user range was resolved via the polar-domain cookbook proposed in \cite{Cui2022channel}.
Further, to break the resolution limitations, an off-grid beam training method was investigated in \cite{wuxun}, which leveraged the range information underlying the energy-spread beam pattern and shared similar beam training overhead with the method in \cite{two_phase}.
However, the beam training overhead of the above two methods is still proportional to the number of antennas, which is prohibitively high in the XL-array scenarios with a large number of antennas.
To resolve this issue, the authors in~\cite{twostage} and \cite{dai_hier} designed an hierarchical beam training scheme, which first estimated a coarse user angle with wide beams and then gradually obtained finer user angle-range points round by round using narrow beams.
Moreover, the authors in \cite{zhou2024multi} extended the multi-beam training scheme in conventional far-field communications to the near-field by utilizing a novel antenna spare-activation method, which resolved several issues such as coverage holes when subarray-based multi-beam generation method is applied in the near-field communications.
In addition, deep learning techniques were utilized in \cite{dl} to reduce the beam training overhead for near-field communications, for which deep neural networks (DNNs) are trained using conventional far-field DFT codebooks and near-field polar-domain codebooks to predict the user angle and range, respectively.
To further improve the prediction accuracy, the authors in \cite{jiang2023near} proposed to train the DNN by utilizing only the near-field polar-domain codebook, which achieves a better achievable-rate performance than other deep-learning based benchmark schemes with the same pilots.
\subsubsection{Wideband Beam Training}
To further reduce beam training overhead, wideband beam training methods were proposed in \cite{cui2022near, zheng2024near}, which essentially leveraged multiple subcarriers to cover different locations or directions in one pilot symbol.
Thanks to the abundant spectrum resources in high frequency bands, wideband beam training methods can estimate the user angle and range within a few pilots.
In particular, the authors in \cite{cui2022near} proposed to control the beam-split effect in near-field wideband communications using the TD beamforming structure and they controlled the beams formed at different subcarriers to sequentially cover specific range rings, indicating that only the overhead of range sampling is incurred.
However, this method still relies on the number of range sampling grids and suffers from a relatively high beam training overhead.
To cope with this problem, a novel near-field wideband beam training method based on the range-dependent beam-split effect was proposed in \cite{zheng2024near}.
This method controlled the beams formed at different subcarriers focused on several inclined strips, which can simultaneously cover multiple angles and range-rings within one pilot, resulting in lower beam training overhead as compared to \cite{cui2022near}.
However, the above two methods can only ensure that the array-gain loss of the optimal beam does not exceed $3$-dB, but they can not overcome the resolution limitations of the angle and range estimation under limited spectral resources.

\subsection{Contributions and Organizations}
To break the resolution limitations, we propose to exploit the beam-split effect in both the frequency and spatial domains to achieve super-resolution beam training for a single-user located in the near-field regions.
Moreover, inspired by the subarray-activation method in~\cite{twostage}, we propose to activate a central subarray to create a far-field scenario, which decouples the angle and range estimation and allows us to complete the beam training with ultra-low overhead.
The main contributions of this paper are summarized as follows.
\begin{itemize}
	\item First, we characterize the beam-split effect in both the frequency and spatial domains for an activated \emph{sparse} uniform linear array. 
	It is shown that the beam pattern of the S-ULA over all subcarriers occurs a new phenomenon called \emph{rainbow blocks}, which can be utilized to extend the beam coverage in the angular domain.
	Moreover, we unveil that different rainbow blocks exhibit different block widths, which increase with the rainbow block index.
	By judiciously designing the TD parameter, we achieve beam coverage over the entire angular domain with an average angular resolution of $ \frac{1}{MU} $, where $ M $ is the number of subcarriers, and $ U $ is the sparse activation interval.
	\item Second, we propose a new three-stage training method with super-resolution and ultra-low overhead based on the controllable beam-split effect in both the frequency and spatial domains of an activated S-ULA. 
%	Specifically, we decouple the angle and range estimation by activating central subarrays including an S-ULA and a dense ULA.
	In the first stage, we activate a central S-ULA and several candidate user angles can be identified via beam sweeping with the rainbow blocks over time.
	Then, in the second stage, we activate a central dense subarray and meticulously select several subcarriers to resolve the angular ambiguity introduced by the beam-split effect in the spatial domain.
	As such, we utilize the single beam formed at each selected subcarrier to precisely cover the candidate user angles and determine the actual user angle by comparing the received power over frequencies.
	In the third stage, the entire XL-array is activated, for which we control the splitted beams over all subcarriers focused within a desired range at the estimated user angle. Consequently, the range estimation can achieve an average resolution of $ M/\Delta_r $, where $ \Delta_r $ denotes the interval of user range distribution.
	\item Finally, extensive numerical results are provided to validate the effectiveness and super-resolution of our proposed three-stage beam training method.
	Compared to benchmark schemes, such as the exhaustive search based and the near-field rainbow based beam training methods, the proposed three-stage beam training method can achieve higher angle and range estimation accuracy with only three pilot symbols.
	This is because much more beams are generated to cover the angular and range domains for achieving higher spatial resolution. 
	Moreover, it is shown that thanks to the super-resolution, the proposed near-field beam training method can achieve nearly the same rate-performance as the perfect CSI based beamforming.	
\end{itemize}

\textit{Organizations}: The rest of this paper is organized as follows.
In Section \ref{Sec:System model}, the channel model and the TD beamforming are introduced.
In Section \ref{Sec:Multi-beam Characteristic}, we discuss the multi-beam characteristic and the rainbow-block distribution for an activated S-ULA.
Then, a three-stage near-field beam training method with super-resolution and ultra-low overhead based on the activated S-ULA is proposed in Section \ref{sec:proposed method}.
Finally, extensive simulation results are presented in section \ref{Sec:numericalResults} to demonstrate the effectiveness of the proposed wideband beam training method, followed by the conclusions made in Section \ref{Sec:Conclusions}. 

\textit{Notations}: Bold lower-case letters represent vectors, while bold upper-case letters denote matrices. In addition, calligraphic letters are used to signify sets.
For vectors and matrices, the symbol $ (\cdot)^{H} $ refers to the conjugate transpose operation. The notations $ \left|\cdot\right| $ and $ \left\lVert \cdot\right\lVert $ correspond to the absolute value of a scalar quantity and the $ \ell_{2} $ norm, respectively.
Finally, the symbol $ \lfloor \cdot \rfloor $ denotes the floor function, which rounds down to the nearest integer, while $ \lceil \cdot \rceil $  represents the ceiling function, which rounds up to the nearest integer.
\begin{figure}
	\centering
	\includegraphics[width = 0.95\columnwidth]{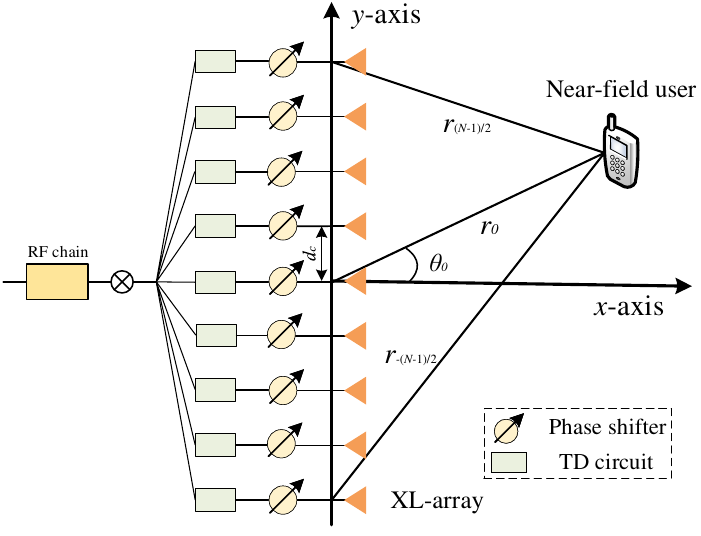}
	\caption{A wideband Near-field XL-array communication system.}
	\label{fig:systemModel}
\end{figure}
\section{System Model} 
\label{Sec:System model}
In this section, we first introduce the near-field wideband channel model for XL-arrays, and then present the TD beamforming architecture, which is utilized to control the wideband beam-split effect.
\subsection{ Near-field Wideband Channel Model}
We consider a wideband XL-array communication system as shown in Fig. \ref{fig:systemModel}, where the BS is equipped with a dense ULA (D-ULA) with $ N $ antennas (assuming to be an odd integer
for convenience), serving a single-antenna user. 

The D-ULA is assumed to be placed along the $y$-axis and centered at the origin. Herein, the $ n $-th antenna of the D-ULA is located at ($0, nd_{\rm c}$), where $ n \in \mathcal{N} \triangleq \{0,\pm 1,\cdots, \pm \frac{N-1}{2}\} $ denotes the antenna index.
The inter-antenna spacing of the D-ULA is $ d_{\rm c} = \frac{\lambda_{\rm c}}{2} $, where $ \lambda_{\rm c} = \frac{c}{f_{\rm c}} $ represents the wavelength of the central subcarrier with $ c $ and $ f_{\rm c} $ denoting the light speed and the frequency of the central subcarrier.
The single-antenna user is located at $ (r_0\sqrt{(1-\theta_0^2)}, r_0\theta_0) $, where $ r_0 $ and $ \theta_0 \in [-1,1) $ represent the BS-user range and spatial angle (see Fig.~\ref{fig:systemModel}), respectively.
The bandwidth is denoted by $ B $, which is divided into $ M $ subcarriers.
In particular, the frequency of the $ m $-th subcarrier is given by
$$ f_m = f_{\rm c} + (m-1-\frac{M-1}{2})\frac{B}{M}, \forall m \in \mathcal{M}, $$
where $ \mathcal{M} \triangleq \{1, 2, \cdots, M\} $ denotes the set of subcarrier index.

For the large-aperture XL-arrays, the user is assumed to be located in its Fresnel near-field region.
In particular, the BS-user range satisfies $ Z_{\rm F} <r_0< Z_{\rm Eff} $, where $Z_{\rm F} = \max{\{ d_{R} , 1.2D \} } $ and $Z_{\rm Eff}=0.367\frac{2D^2}{\lambda_{\rm c}}$ denote the Fresnel distance and effective Rayleigh distance with $ D = (N-1)d_{\rm c} $ denoting the array aperture \cite{cui2024near}.
Herein, $ d_{R} $ is shown to be several wavelengths in \cite{ouyang2024impact} and the Fresnel distance can be simplified by $ Z_{\rm F} = 1.2D $.
For example, when the central carrier operates at $f_{\rm c} = 60$ GHz and the number of antennas is $N = 513$, the effective Rayleigh distance $ Z_{\rm R} $ is $120$ m, which means that the user is more likely to be located in the near-field region of the XL-array.
Therefore, the channel between the BS and the user can be modeled by the uniform spherical wavefronts (USW)~\cite{lu2023tutorial}.
Moreover, due to severe path-loss and shadowing effects \cite{nf_ris} in high-frequency bands such as mmWave and even THz, only the LoS path is considered in this paper.
As such, the near-field LoS channel from BS$\to$user can be modeled as~\cite{yunpu_swipt}
\begin{equation}\label{Eq:nf-model}
	\mathbf{h}^H_{m} \approx \sqrt{N}\beta_m e^{-\frac{\jmath 2 \pi}{\lambda_m}r_0}\mathbf{b}_m^{H}(r_{0}, \theta_{0}),
\end{equation}
%Then, the multi-path channel between the XL-array and the single user at the $ m $-th subcarrier can be modeled as \cite{yunpu_swipt}
%\begin{equation}
%	\label{Eq:nf-channel}
%	\mathbf{h}^H_{m} \!=\! \sqrt{N}\beta_m \mathbf{b}_m^{H}(r_{0}, \theta_{0}) \!+\! \sum_{\ell=1}^{L_m} \sqrt{\frac{N}{L_m}}\beta_{m,\ell} \mathbf{b}_m^{H}({r}_{m, \ell}, {\theta}_{m, \ell}),
%\end{equation}
%which involves one LoS path and $ L_m $ non-LoS (NLoS) paths.
where the parameter  $ \beta_m = \frac{\lambda_m}{4\pi r_0} $ represents the LoS path gain at the $ m $-th subcarrier with $ \lambda_m $ denoting the wavelength of the $ m $-th subcarrier.
In particular, $\mathbf{b}_m(r_{0},\theta_{0})$ represents the near-field channel response vector, given by \cite{yunpu2} 
\begin{equation}\label{eq:near_steering}
	\left[\mathbf{b}_m^H\left(r_{0},\theta_{0}\right) \right]_{n} = \frac{1}{\sqrt{N}}e^{-\frac{\jmath 2 \pi}{\lambda_m}(r_{n}-r_0)}, \forall n\in \mathcal{N},
\end{equation}
with $r_{n} =\sqrt{r_{0}^2+ {n}^2d_c^2-2r_{0}\theta_{0} nd_c}$ denoting the range between the $ n $-th antenna and the user. 
Since $ r_{n} $ is a complicated radical function and difficult to be analyzed.
To resolve this issue, the Fresnel approximation can be utilized, which is shown to be accurate in the Fresnel region~\cite{liu_review}.
Then, $ r_n $ can be approximated as \cite{shao2022target}
\begin{equation}
\label{eq:Fresnel approximation}
	r_{n} \approx r_{0}-nd_c \theta_{0}+n^2 d_c^2\mu_0 ,
\end{equation}
where $ \mu_0 =  \frac{1- \theta_{0}^2}{2 r_{0}} $ \cite{zhou2024sparse}.
\vspace{-10pt}
\subsection{Time-delay and Phase-shifter Based Beamforming} 
Since the channel model in \eqref{Eq:nf-model} is frequency-dependent, when only the frequency-independent PS-based beamforming is applied, the beam-split effect may significantly degrade the rate-performance of wideband communication systems.
To tackle this problem, we consider the TD-based beamforming in this paper, which forms frequency-dependent beams and hence becomes an effective method to compensate or control the near-field beam-split effect~\cite{liao2021terahertz, tan2021wideband, lin2017subarray}. 

Specifically, the XL-array is attached with $ N $ TD circuits and $ N $ PSs, for which each antenna is connected to one TD circuit and one PS.
Each TD circuit is able to tune frequency-dependent phase shifts by introducing controllable delays on the wideband signals.
Mathematically, the TD beamformer for the XL-array is denoted by $ \mathbf{w}_m \in \mathbb{C}^{N \times 1} $, which is given by~\cite{chen2024near}
$$ [\mathbf{w}_m^H]_n = \frac{1}{\sqrt{N}}e^{-{\jmath 2 \pi f_m \tau_n}}, \forall n\in \mathcal{N}, $$
where $ \tau_n $ denotes the adjustable delay of the $ n $-th TD circuit.
In particular, to match the near-field channel \cite{cui2022near}, we set $ \tau_n = \frac{nd_c\theta^{\prime} - n^2d_c^2\mu^{\prime}}{c} $ with $ \theta^{\prime} $ and $ \mu^{\prime} $ denoting the adjustable TD angle and range parameters.
Then, the TD beamformer can be rewritten as
\begin{equation}
	[\mathbf{w}_m^H(\theta^{\prime}, \mu^{\prime})]_n = \frac{1}{\sqrt{N}}e^{-{\jmath \frac{2 \pi}{\lambda_m} (nd_{\rm c}\theta^{\prime} - n^2d_c^2\mu^{\prime})}}, \forall n\in \mathcal{N},
\end{equation}
which is similar to the form of the channel response vector in~\eqref{eq:near_steering}.
Similar to the TD beamforming, the PS beamformer for the XL-array is given by
\begin{equation}
	\left[\mathbf{w}_{\mathrm{PS}}\left(\theta_p^{\prime}, \mu_p^{\prime}\right)\right]_n=\frac{1}{\sqrt{N_t}} e^{-\jmath \frac{2\pi}{\lambda_c}\left(n d \theta_p^{\prime}-n^2 d^2 \mu_p^{\prime}\right)}, \forall n\in \mathcal{N},
\end{equation}
where $ \theta_p^{\prime} $ and $ \mu_p^{\prime} $ denote the PS angle and range parameters, respectively.
Based on the above, the effective beamformer $ \mathbf{w}_m(\theta^{\prime}, \mu^{\prime},\theta_p^{\prime},\mu_p^{\prime}) $ for the XL-array can be written as
\begin{equation}
	\mathbf{w}_m(\theta^{\prime}, \mu^{\prime};\theta_p^{\prime},\mu_p^{\prime}) = \mathbf{w}_m(\theta^{\prime}, \mu^{\prime}) \odot \mathbf{w}_{\mathrm{PS}}\left(\theta_p^{\prime}, \mu_p^{\prime}\right).
\end{equation}
Then, the received signal at the $ m $-th subcarrier by the user is given by
\begin{equation}
	y_{m} = \sqrt{P_t} \mathbf{h}_m^H \mathbf{w}_m(\theta^{\prime}, \mu^{\prime}; \theta_p^{\prime}, \mu_p^{\prime}) x_m + n_m,
\end{equation}
where $ P_t $ and $ x_m $ denote the transmit power and the baseband signal with $ \mathbb{E}(|x_m|^2) = 1  $. Moreover, $ n_m \sim \mathcal{C N}(0, \sigma^2) $ represents the additive white Gaussian noise (AWGN) with $ \sigma^2 $ denoting the noise power.
\section{Multi-beam Characteristic of Activated Central S-ULA}
\label{Sec:Multi-beam Characteristic}
%In this section, we propose to leverage the multi-beam characteristic of an activated central S-ULA, which assists to achieve the super-resolution wideband beam training with limited frequency and time resources.
In this section, we characterize the beam property of a sparse array in the wideband\footnote{Considering that the TD beamforming is enough for the angle control of the beams at entire subcarriers, the PSs are turned off for saving energy.}. 
Interestingly, it is shown that the sparse array and wide bands both lead to the beam-split effect, corresponding to the spatial and frequency domains, respectively.
Such beam-split effects in both spatial and frequency domains will be exploit in Section \ref{sec:proposed method} to design an efficient wideband beam training method for near-filed communications.
\vspace{-10pt}
\subsection{Far-field Wideband Channel Model for S-ULA}
\label{sec:far-field channel}
First, we choose an central subarray with $ Q $ antennas, for which the user is assumed to be located in the far-field region of the subarray\footnote{Indeed, this condition is easily satisfied with a considerable number of antennas. For example, assuming $ r_0 > 10 $ m, the number of antennas can be chosen as $ Q = 75 $.}, with respect to (w.r.t.) $ r_0 > R_{\rm Eff}^{\rm SA} = 0.367\frac{2(Q-1)^2d_{\rm c}^2}{\lambda_{\rm c}} $.
Then, we uniformly activate $ \widetilde{Q} $ antennas of the central subarray with $ (U-1) $ deactivated antennas in between as shown in Fig. \ref{fig:sparse array}.
In particular, we have $ \widetilde{Q} = \frac{Q-1}{U} + 1 $, which is assumed to be an integer.
%\footnote{Indeed, the antennas not utilized within the beam training stage can be activated to serve specific users, thereby reducing the waste of hardware resources.}
As such, the central subarray with $ Q $ antennas becomes an S-ULA with $ \widetilde{Q} $ antennas.
Hence, the LoS channel between the S-ULA and the user at the $ m $-th subcarrier can be modeled under the conventional planar wavefront, given by
\begin{equation}\label{Eq:SA-ff-model}
	(\mathbf{h}^{\rm SA}_{m})^H = \sqrt{\widetilde{Q}}\beta_m \mathbf{a}_m^{H}(\theta_{0}, U),
\end{equation}
where $ \mathbf{a}_m(\theta_{0}, U) $ represents the far-field array response vector for the activated S-ULA.
Mathematically, $ \mathbf{a}_m(\theta_{0}, U) $ can be expressed as
\begin{equation}\label{SA-far_steering}
	\left[\mathbf{a}_m^H(\theta_{0}, U) \right]_{\tilde{q}} = \frac{1}{\sqrt{\widetilde{Q}}}e^{\jmath\frac{ 2 \pi }{\lambda_m}\tilde{q}Ud_{\rm c}\theta_{0}}, \forall \tilde{q}\in \widetilde{\mathcal{Q}},
\end{equation}
where $ \tilde{q} \in \widetilde{\mathcal{Q}} \triangleq \{-\frac{\widetilde{Q}-1}{2}, -\frac{\widetilde{Q}-1}{2}+1, \cdots, \frac{\widetilde{Q}-1}{2}\} $ denotes the antenna index in the activated S-ULA.
Herein, the far-field TD beamformer for the S-ULA is given by
\begin{figure}
	\centering
	\includegraphics[width = 0.85\columnwidth]{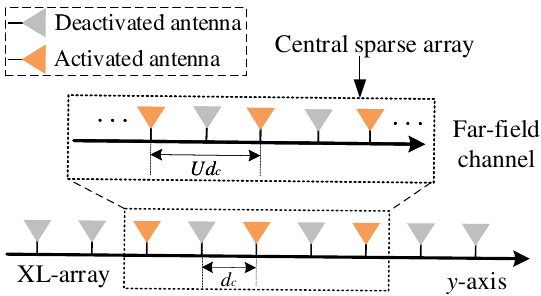}
	\caption{Illustration of an activated central sparse ULA.}
	\label{fig:sparse array}
\end{figure}
\begin{equation}
	\label{eq:SA TD beamformer}
	[\mathbf{w}_m^H(\theta^{\prime}_{\rm SA}, U)]_{\tilde{q}} = \frac{1}{\sqrt{\widetilde{Q}}}e^{-{\jmath \frac{2 \pi}{\lambda_m} \tilde{q}Ud_{\rm c}\theta^{\prime}_{\rm SA}}}, \forall \tilde{q}\in \widetilde{\mathcal{Q}},
\end{equation}
where $ \theta_{\rm SA}^{\prime} = \frac{{\tilde{q}}Ud_{\rm c}\tau_{\tilde{q}}}{c} $ denotes the adjustable TD parameter, while the actual delay of the $ \tilde{q} $-th TD circuit in the activated S-ULA can be represented as $ \tau_{\tilde{q}} = \frac{{\tilde{q}}Ud_{\rm c}\theta_{\rm SA}^{\prime}}{c} $.
\subsection{Multi-beam Characteristic of S-ULA}
Let $ f_m(\theta,\theta^{\prime}_{\rm SA}, U) $ denote the array gain of the activated S-ULA at an observation angle $ \theta $, for which the TD beamformer is set as $ \mathbf{w}_m(\theta^{\prime}_{\rm SA}, U) $ in \eqref{eq:SA TD beamformer}. As such, $ f_m(\theta,\theta^{\prime}_{\rm SA}, U) $ is given by \cite{zhou2024multi}
\begin{equation}
\label{eq:SA array gain}
	\begin{aligned}
		f_m(\theta,\theta^{\prime}_{\rm SA}, U) =\left|\mathbf{w}_m^H(\theta^{\prime}_{\rm SA}, U) \mathbf{a}_m(\theta, U)\right|.
	\end{aligned}	
\end{equation}
Based on \eqref{eq:SA array gain}, we have the following results.
\begin{lemma}[Multi-beam characteristic]
	\label{lemma:multi-beam}
	\emph{Given an activated S-ULA parameterized by $ U	$ and a TD beamformer $ \mathbf{w}_m(\theta^{\prime}_{\rm SA}, U) $. $ |\mathcal{K}_m| $ beams will be formed at the $ m $-th subcarrier, for which the multi-beam beam angle $ \theta_{m}^{(u_m)} $ is given by}
	$$ \theta_m^{(u_m)} = \theta^{\prime}_{\rm SA} + \frac{2k_m^{(u_m)}}{U\rho_m}, \forall u_m \!=\! 1,2, \cdots, |\mathcal{K}_m|, $$
	\emph{where $ \rho_m =  \frac{f_m}{f_c}$ {and} $ k_{m}^{(u_m)} \in \mathcal{K}_m \triangleq \Big\{ \mathbb{Z} ~ \cap \big\{k| \theta^{\prime} + \frac{2k}{U\rho_m} \in [-1,1) \big\} \Big\} $ with $ \mathbb{Z} $ denoting the integer set\footnote{Without loss of generality, we typically assume that $ k_{m}^{(u_m+1)} > k_m^{(u_m)}, \forall k_m^{(u_m)}, k_{m}^{(u_m + 1)} \in \mathcal{K}_m $, leading to the fact that $ k_{m}^{(u_m+1)} = k_m^{(u_m)} + 1 $.}.}
%	\emph {Moreover, $ u_m \in \{1, 2, \cdots, |\mathcal{U}_m| \}$ denotes the multi-beam index, satisfying $ k_{u_m} = k_{(u_m-1)} + 1 $.}
\end{lemma}
\begin{proof}
	The array gain in \eqref{eq:SA array gain} can be simplified as
	\begin{equation}
	\begin{aligned}
		f_m(\theta,\theta^{\prime}_{\rm SA}, U) &= \frac{1}{\widetilde{Q}}\Big|\sum_{\tilde{q}\in \widetilde{\mathcal{Q}}} e^{\jmath \frac{2\pi}{\lambda_m}\tilde{q} Ud_{\rm c}(\theta- \theta^{\prime}_{\rm SA})}\Big|\\
		&\overset{(a_1)}{=} \frac{\sin \frac{\pi\rho_m}{2}\widetilde{Q}U(\theta - \theta^{\prime}_{\rm SA}) }{\widetilde{Q}\sin \frac{\pi\rho_m}{2}U(\theta - \theta^{\prime}_{\rm SA})},
	\end{aligned}
	\end{equation}
	where $ (a_1) $ is obtained from \cite{zhou2024sparse}. By setting $ \frac{\pi\rho_m}{2}U(\theta - \theta^{\prime}_{\rm SA}) = k_{m}\pi $ w.r.t. $ \theta_{m}^{(u_m)} = \theta^{\prime}_{\rm SA} + \frac{2k_{m}^{(u_m)}}{U\rho_m} $, where $ k_{m}^{(u_m)} \in \mathbb{Z} $, we have $ f_m(\theta,\theta^{\prime}_{\rm SA}, U) = 1 $.
	Considering that $ \theta_{m}^{(u_m)} $ is an actual spatial angle satisfying $ -1 \le \theta_{m}^{(u_m)} < 1 $, $ k_{m}^{(u_m)} $ is constrained by 
	$$ k_{m}^{(u_m)} \in \mathcal{K}_m \triangleq \Big\{ \mathbb{Z} ~ \cap \big\{k| \theta^{\prime}_{\rm SA} + \frac{2k}{U\rho_m} \in [-1,1) \big\} \Big\}. $$
	Hence, $ |\mathcal{K}_m| $ beams are formed at the $ m $-th subcarrier in the angular domain and we thus complete the proof.
\end{proof}

From Lemma \ref{lemma:multi-beam}, it is observed that the periodicity of the multiple beams formed at the $ m $-th subcarrier is $ \frac{2}{U\rho_m} $.
Therefore, $ |\mathcal{K}_m| = \lfloor U\rho_m \rfloor  $ or $ \lfloor U\rho_m \rfloor + 1 $, depending on the adjustable TD parameter $ \theta^{\prime}_{\rm SA} $.
In particular, the beam periodicity at the central carrier is $ \frac{2}{U} $, resulting in $ 2/(2/U) = U $ beams in the angular domain.
Overall, the number of multiple beams $ |\mathcal{K}_m| $ depends on the TD parameter $ \theta^{\prime}_{\rm SA} $ and the subcarrier frequency $ f_m $, which is further presented below.
\begin{lemma}[Number of splitted beams at subcarrier $ m $]
	\label{lemma:number of multi-beams}
	\emph{Given an S-ULA with the parameter $ U $ and a TD beamformer $ \mathbf{w}_m(\theta^{\prime}_{\rm SA}, U) $, the number of splitted beams formed at the $ m $-th subcarrier is given by}
	\begin{equation}
		\begin{aligned}
		\label{eq:nuber of multi-beama}
			|\mathcal{K}_m| =
			\left\{\begin{matrix}
				\lfloor U\rho_m \rfloor, \quad\quad~\theta^{\prime}_{\rm SA} + \frac{k_m^{(1)}}{U\rho_m} \ge 1 - \frac{2\lfloor U\rho_m \rfloor}{ U\rho_m}, \\ 
				\lfloor U\rho_m \rfloor + 1, ~~\theta^{\prime}_{\rm SA} + \frac{k_m^{(1)}}{U\rho_m} < 1 - \frac{2\lfloor U\rho_m \rfloor}{ U\rho_m}.
		\end{matrix}\right.		
		\end{aligned}
	\end{equation}
\end{lemma}
\begin{proof}
	For the case $ \theta^{\prime}_{\rm SA} + \frac{k_m^{(1)}}{U\rho_m} = \theta_m^{(1)} < 1 - \frac{2\lfloor U\rho_m \rfloor}{U\rho_m} $, we can obtain
	$$ -1 \!\le\! \theta_m^{(u_m)} \!=\! \theta_m^{(1)} \!+\! \frac{2(u_m-1)}{U\rho_m} \!<\! 1, u_m \!=\! 1,2,\cdots, \lfloor U\rho_m \rfloor \!+\! 1, $$
	which indicates that $ \lfloor U\rho_m \rfloor + 1 $ beams are generated at the $ m $-th subcarrier.
	For the other case $ \theta^{\prime}_{\rm SA} + \frac{k_m^{(1)}}{U\rho_m} \ge 1 - \frac{2\lfloor U\rho_m \rfloor}{U\rho_m} $, we have 
	\begin{align}
		\theta_m^{(\lfloor U\rho_m \rfloor + 1)} &= \theta_m^{(1)} + \frac{2\lfloor U\rho_m \rfloor}{U\rho_m} \nn\\
		&\ge 1 - \frac{2\lfloor U\rho_m \rfloor}{ U\rho_m} + \frac{2\lfloor U\rho_m \rfloor}{ U\rho_m} = 1,
	\end{align}
	which means that only $ \lfloor U\rho_m \rfloor $ beams are formed at the $m$-th subcarrier. Thus, we have completed the proof.
\end{proof}
\begin{remark}
	\emph{Based on Lemma \ref{lemma:number of multi-beams} and the condition $ B \ll f_c $ (i.e., $ \rho_m \approx 1 $), we have 
	\begin{equation}
		\lfloor U\rho_m \rfloor = \left\{\begin{matrix}
			U, \quad\quad~f_m \ge f_{\rm c}, \\ 
			U - 1, ~~f_m < f_{\rm c}.
		\end{matrix}\right.
	\end{equation}
	Hence, the number of splitted beams formed at the subcarriers with  $ f_m > f_{\rm c}  $ is $ U $ or $ U + 1 $, while the number of multiple beams generated at the subcarriers with $ f_m < f_{\rm c} $ is $ (U-1) $ or $ U $. 
 	Overall, approximately $ MU $ beams are simultaneously generated across the $ M $ subcarriers in the angular domain.
 	By effectively leveraging these beams for angle sweeping, we can achieve higher angular resolution, as compared to conventional arrays with half-wavelength spacing given the same frequency resources.} 
\end{remark}
%Although the selection of TTD parameters results in a difference of only one beam at each subcarrier, which is negligible in terms of beam numbers, the number of multi-beams at different subcarriers significantly affects the subsequent resolution of angular ambiguity.
%Specifically, we utilize the multi-beam characteristic to scan the angular domain, which introduces angular ambiguity.
%To resolve this ambiguity, $ |\mathcal{K}_m| $ single beams are sequentially formed at $ |\mathcal{K}_m| $ subcarriers to cover the candidate angles.
%Hence, it is necessary to investigate the impact of TD parameters and frequencies on the number of formed beams at different subcarriers.
\begin{figure}
	\centering
	\vspace{-14pt}
	\includegraphics[width = 0.85\columnwidth]{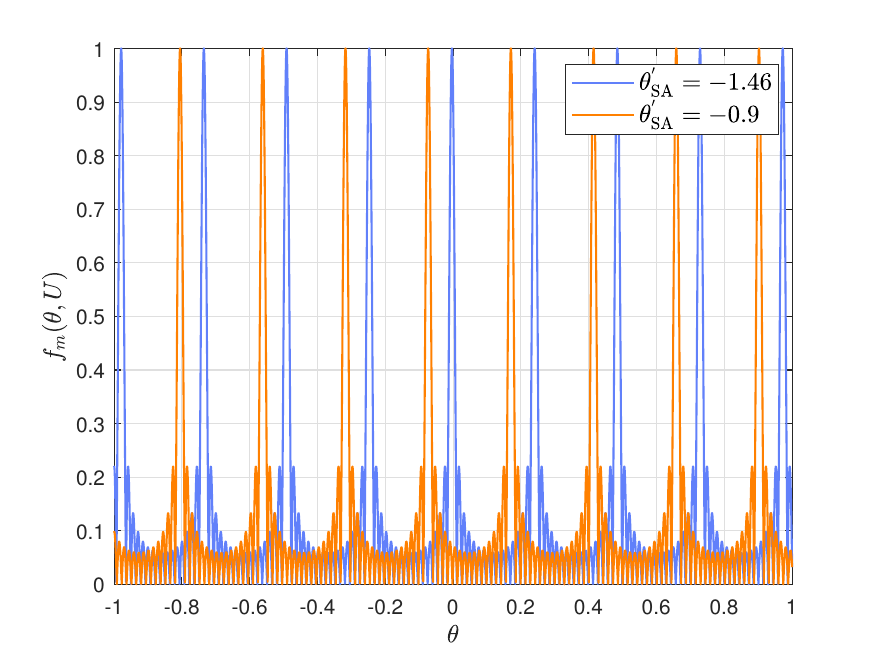}
	\caption{Different numbers of splitted beams in the spatial domain with different TD parameters $ \theta_{\rm SA}^{\prime} $, where $ f_{\rm c} = 30 $ GHz, $ B = 3 $ GHz and $ M = 1024 $.}
	\label{fig:different number of beams}
	\vspace{-10pt}
\end{figure}
\begin{example}
	\emph{We plot the array gain at the $ M $-th frequency versus the spatial angle in Fig. \ref{fig:different number of beams}.
	We suppose that a central S-ULA with $ \widetilde{Q} = 17 $ antennas are activated and the activation interval is $ U = 8 $.
	The BS is assumed to be operated at the central frequency $ f_{\rm c} = 60 $ GHz and the bandwidth is $ B = 3 $ GHz with $ M = 1024 $ subcarriers.
	We consider the multi-beam of the frequency $ f_{\rm H} = f_M =  61.5$ GHz.
	\begin{itemize}
		\item For the TD parameters $ \theta^{\prime}_{\rm SA} = -1.46$ with $ k_M^{(1)} = 2 $, we have $ \theta^{\prime}_{\rm SA} + \frac{k_M^{(1)}}{U\rho_m} = -1.22 < 1 - \frac{2\lfloor U\rho_M \rfloor}{ U\rho_M} = -0.95 $. Hence, $ \lfloor U\rho_M \rfloor + 1 = 9 $ beams can be observed in Fig.~\ref{fig:different number of beams}.
		\item For the TD parameters $ \theta^{\prime}_{\rm SA} = -0.9$ with $ k_M^{(1)} = 0 $, we have $ \theta^{\prime}_{\rm SA} + \frac{k_M^{(1)}}{U\rho_m} = -0.9 > 1 - \frac{2\lfloor U\rho_M \rfloor}{ U\rho_M} = -0.95 $, resulting in $ \lfloor U\rho_M \rfloor = 8 $ beams formed at the $ M $-th subcarrier, as illustrated in Fig.~\ref{fig:different number of beams}.
	\end{itemize}
	This example validates the analysis for the number of splitted beams at different subcarriers in \eqref{eq:nuber of multi-beama}.}	
\end{example} 

\subsection{Beam Pattern of All Subcarriers}
\label{sec:normal TD parametr}
To achieve controllable beam-split effects, we analyze the beam pattern of all subcarriers under the following two conditions: 1)  $ \theta^{\prime}_{\rm SA} \in [-1,1) $; 2) $ \theta^{\prime}_{\rm SA} \notin [-1,1) $.
\begin{lemma}
	\label{lemma:ineffective TD parameter}
	\emph {When the TD parameter satisfies $ \theta^{\prime}_{\rm SA} \in [-1,1) $, there exists one beam formed at each subcarrier that steers towards the same angle $ \theta_m^{(u_m)} = \theta^{\prime}_{\rm SA} $.}
\end{lemma}
\begin{proof}
	When $ k_m^{(u_m)} = 0, \forall m \in \mathcal{M} $, we have $ \theta_m^{(u_m)} = \theta^{\prime}_{\rm SA} \in [-1,1) $, which is a physical angle, resulting in $ k_m^{(u_m)} = 0 \in \mathcal{K}_m, \forall m \in \mathcal{M}$.
	Hence, there exists one beam formed at each subcarrier steered towards the same angle $ \theta_m^{(u_m)} = \theta^{\prime}_{\rm SA} $.
\end{proof}
\begin{figure}
	\centering
	\includegraphics[width = 1\columnwidth]{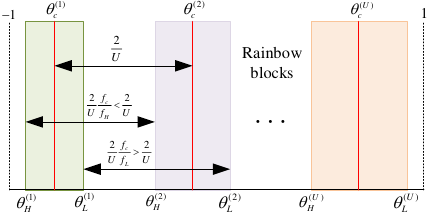}
	\caption{Illustration of the multi-beam beam distributions and far-field rainbow blocks.}
	\label{fig:rainbow blocks}
\end{figure}

Lemma \ref{lemma:ineffective TD parameter} indicates that when $ \theta^{\prime}_{\rm SA} \in [-1,1) $, $M$ beams are not effectively utilized to cover different angles, which is not the desired scheme in the subsequent beam training design.
This case is thus referred to as an \emph{ineffective} TD parameter condition.
Therefore, we mainly consider the case $ \theta^{\prime}_{\rm SA} \notin [-1,1) $ in this paper.
In this case, we can control all the beams over the subcarriers to achieve super-resolution angle estimation.
Next, we mainly focus on the case $ \theta^{\prime}_{\rm SA} < -1 $, while similar analysis can be obtained for the case $ \theta^{\prime}_{\rm SA} > 1 $. 
For $ \theta^{\prime}_{\rm SA} < -1 $, we have $ k_m^{(u_m)} > 0, \forall k_m^{(u_m)} \in \mathcal{K}_m, \forall m \in \mathcal{M}. $

Then, we present the multi-beam pattern of all subcarriers in detail under the condition $ \theta^{\prime}_{\rm SA} < -1 $.
We first define the concept of a \emph{rainbow block} as follows.
\begin{definition}[Rainbow block]
	\label{definition:rainbow block}
	\emph{For each subcarrier, we collect one of the multi-beam angles that shares the same $k_m^{(u_m)}$ as the central subcarrier $k_c^{(u_{\rm c})}$ into a set. This set is referred to as a \emph{rainbow block} w.r.t. the parameter $k_c^{(u_{\rm c})}$.
	In particular, since $ |\mathcal{U}_{\rm c}| = U $, there exist $ U $ rainbow blocks.
	Mathematically, the $ u_{\rm c} $-th rainbow block $ \mathcal{A}^{(u_{\rm c})} $ is defined as
	\begin{equation}
		\mathcal{A}^{(u_{\rm c})} \triangleq \{\theta|\theta = \theta^{\prime}_{\rm SA} + \frac{2k_{\rm c}^{(u_{\rm c})}}{U\rho_m}, \forall m \in \mathcal{M}\},
	\end{equation}
	where $ u_{\rm c} \in \mathcal{U}_{\rm c} \triangleq \{1,2,\cdots U\} $. 
	}	
\end{definition}
To this end, we define the angular interval $ [\theta^{\prime}_{\rm SA} + \frac{2k_{\rm c}^{(u_{\rm c})}}{U\rho_{\rm H}},\theta^{\prime}_{\rm SA} + \frac{2k_{\rm c}^{(u_c)}}{U\rho_{\rm L}}] $ as the coverage region of the $ u_{\rm c} $-th rainbow block\footnote{Considering wideband communication systems, where the number of subcarriers typically exceeds the number of antennas, we can always ensure $ 3 $-dB beam-width within the coverage region of the $ u_{\rm c} $-th rainbow block.}. 
It is worth noting that for the first and last rainbow blocks, there may exist subcarriers $ f_{m_\ell} > f_{\rm c} $ and $ f_{m_r} < f_{\rm c} $ such that
$$ \theta_{m_\ell}^{(1)} = \theta^{\prime}_{\rm SA} + \frac{2k_{\rm c}^{(1)}}{U\rho_{m_{\ell}}} < -1,~~	\theta_{m_r}^{(U)} = \theta^{\prime}_{\rm SA} + \frac{2k_{\rm c}^{(U)}}{U\rho_{m_{r}}} > 1. $$

In other words, the actual angle $ \theta_{m_\ell}^{(1)} $ is given by $ \theta_{m_\ell}^{(1)} = \theta^{\prime}_{\rm SA} + \frac{2k_{m_\ell}^{(1)}}{U\rho_{m_{\ell}}} $ with $ k_{m_\ell}^{(1)} = k_{\rm c}^{(1)}+1 = k_{\rm c}^{(2)} $, while the actual angle $ \theta_{m_r}^{(U)} $ does not exist.
However, this case will not affect our subsequent beam training scheme design.
Without loss of generality, we choose to reuse $ \theta^{(u_{\rm c})}_m $ to denote the angle $ \theta^{\prime}_{\rm SA} + \frac{2k_{m}^{(u_{\rm c})}}{U\rho_{m_{\ell}}} $ in the following discussions, regardless of whether it represents a physical angle or not. 
Then, based on Definition~\ref{definition:rainbow block}, several key observations can be made as follows.
\begin{itemize}
	\item There are $ U $ rainbow blocks and each contains $ M $ angles formed at different subcarriers.
	\item The angles in each rainbow block decrease with the frequency. Specifically, we have $ \theta_m^{(u_{\rm c})} < \theta_{\rm c}^{(u_{\rm c})}, \forall \theta_m^{(u_{\rm c})} \in  \mathcal{A}^{(u_{\rm c})} $ when $ f_m > f_{\rm c} $, while $ \theta_m^{(u_{\rm c})} > \theta_{\rm c}^{(u_{\rm c})} $ when $ f_m < f_{\rm c} $.
	In other words, we can approximately regard the beam angle formed at the central subcarrier as the center of each angular rainbow block.
\end{itemize}

To analyze the beam properties of rainbow blocks, several definitions are presented below.
\begin{definition}[Rainbow block width]
	\emph{Given the TD parameter $ \theta^{\prime}_{\rm SA} $, the width of the $ u_{\rm c} $-th rainbow block is defined as the coverage region of the $ u_{\rm c} $-th rainbow block. Mathematically, we have
		\begin{equation}
			{\rm RW}^{(u_{\rm c})} =  \theta_{\rm L}^{(u_{\rm c})} - \theta_{\rm H}^{(u_{\rm c})} = \frac{2f_cBk_c^{(u_{\rm c})}}{Uf_Lf_H}.
	\end{equation}}
\end{definition}
\begin{definition}[Inter-rainbow interval]
	\emph{Given the TD parameter $ \theta^{\prime}_{\rm SA} $, the inter-rainbow interval between the $ u_{\rm c} $-th rainbow block and $ (u_{\rm c}+1) $-th rainbow block is defined as
	\begin{equation}
		\label{eq:rainbow block gap}
		\begin{aligned}
			{\rm RG}^{(u_{\rm c})} = \theta_{\rm H}^{(u_{\rm c}+1)} - \theta_{\rm L}^{(u_{\rm c})} &= \frac{2k_{\rm c}^{(u_{\rm c}+1)}}{U\rho_{\rm H}} - \frac{2k_{\rm c}^{(u_{\rm c})}}{U\rho_{\rm L}} \\ &= \frac{2f_{\rm c}}{Uf_{\rm L}f_{\rm H}}[-Bk_{\rm c}^{(u_{\rm c})} + f_{\rm L}].
		\end{aligned}	
	\end{equation}
}
\end{definition}

Based on the above definitions, it is observed that the beam pattern of all subcarriers mainly consists of $ U $ rainbow blocks.
Moreover, the width of the rainbow blocks increases with the rainbow block index $ u_{\rm c} $, as illustrated in Fig.~\ref{fig:rainbow blocks}.
In addition, from \eqref{eq:rainbow block gap}, the inter-rainbow interval $ {\rm RG}^{(u_{\rm c})} $ decreases with $ k_{\rm c}^{(u_{\rm c})} $, due to $ -B < 0 $.
When the inter-rainbow interval $ {\rm RG}^{(u_{\rm c})} < 0 $, we refer to this case as achieving \emph{seamless} beam coverage over the angular region of $ [\theta_{\rm H}^{(u_{\rm c})}, \theta_{\rm L}^{(u_{\rm c}+1)}] $.
Next, we discuss how to leverage the beams in the $ U $ rainbow blocks to seamlessly cover the entire angular domain.

First, we let the central subcarrier steer $ U $ beams towards the angles $ \theta_{\rm c}^{(u_{\rm c})} = -1 + \frac{2u_{\rm c}-1}{U}, \forall u \in \mathcal{U}_{\rm c} $, resulting in the following constraint
\begin{equation}
\label{eq:condition f_c}
	\theta_{\rm SA}^{\prime} + \frac{2k_{\rm c}^{(1)}}{U}= -1 + \frac{1}{U}.
\end{equation}
\begin{figure*}[t]
	\centering
	\vspace{-14pt}
	\subfloat[Multi-beam distribution within the rainbow blocks.]{
		\includegraphics[width=0.45\textwidth]{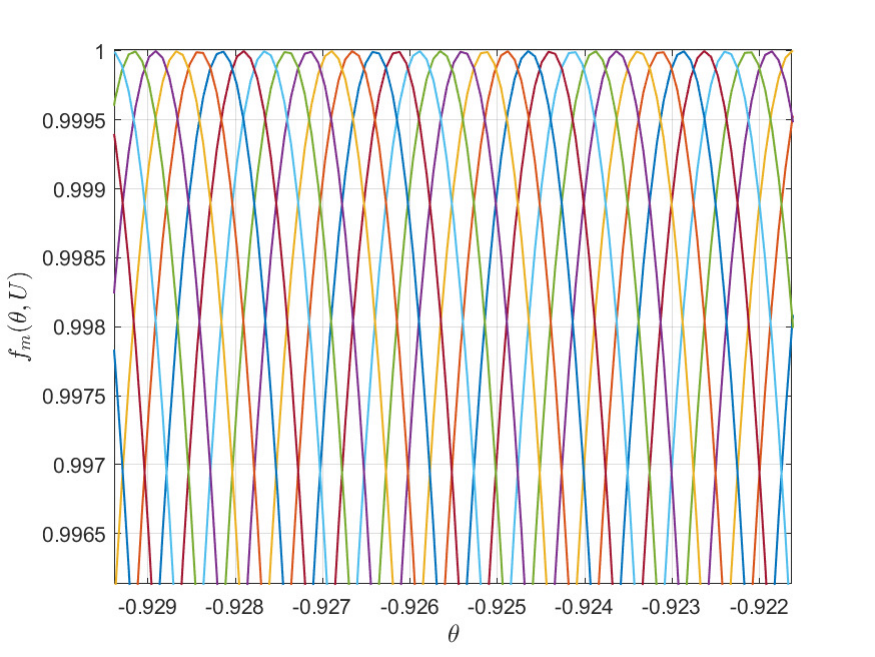}}
	\subfloat[Multi-beam distribution in the boundaries.]{
		\includegraphics[width=0.45\textwidth]{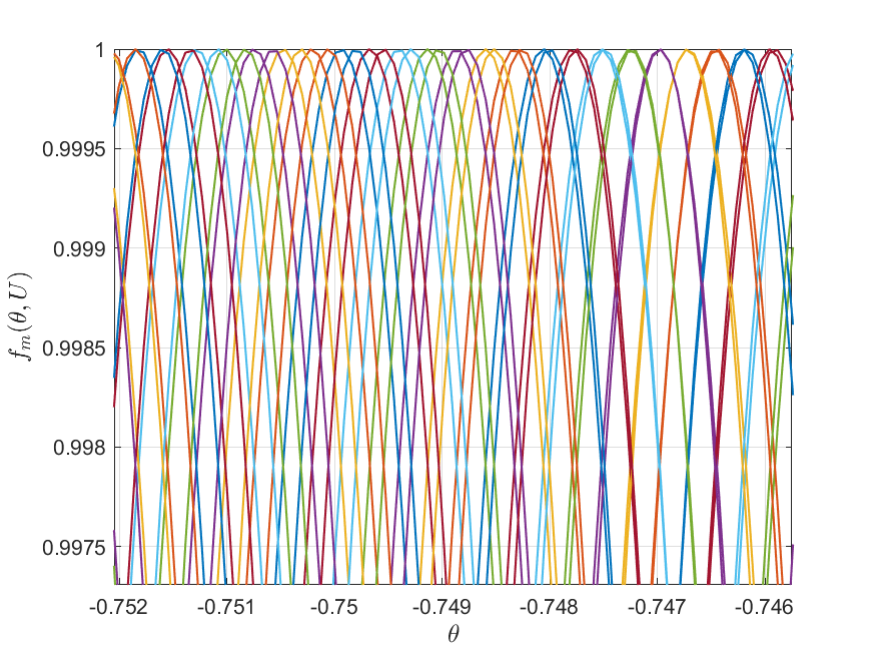}}
	\caption{Multi-beam distributions, where $ f_{\rm c} = 60 $ GHz, $ B = 3 $ GHz and $ M = 1024 $.}
	\vspace{-10pt}
	\label{fig:coverage}
\end{figure*}Then, we impose the following two conditions to seamlessly cover the entire angular domain $ [-1,1) $.
\begin{itemize}
	\item Condition 1: The inter-rainbow interval between each two rainbow blocks is smaller than zero, i.e., $ {\rm RG}^{(u_{\rm c})} \le 0, \forall u_{\rm c} \in \mathcal{U}_{\rm c} $.
	\item Condition 2: The left edge of the first rainbow block $ \theta_{\rm H}^{(1)} $ covers -1, while the right edge of the last rainbow block $ \theta_{\rm L}^{(U)} $ covers 1.
\end{itemize}

Condition 1 means that the angular range $ [\theta_{\rm H}^{(1)}, \theta_{\rm L}^{(U)}] $ is seamlessly covered by the beams within the $ U $ rainbow blocks, while the second condition further ensures that the entire angular domain $ [-1,1) $ is covered.
%Due to the wider rainbow-block inter-rainbow interval between the fist and second rainbow block, if $ {\rm RG}^{(1)} \le 0 $, we have $ {\rm RG}^{(u_c)} \le 0, \forall u_c \in \mathcal{U}_c $.
Since the inter-rainbow interval $ {\rm RG}^{(u_{\rm c})} $ decreases as $ k_{\rm c}^{(u_{\rm c})} $ increases, the first condition results in the following constraint
\begin{equation}
	\label{eq:condition 3}
	{\text{(Constraint 1)}}~~{\rm RG}^{(1)} \le 0.
\end{equation}
Then, Condition 2 imposes constraints on
\begin{equation}
\label{eq:condition 1}
	{\text{(Constraint 2)}}~~\theta_{\rm H}^{(1)} = \theta_{\rm SA}^{\prime} + \frac{2k_{\rm c}^{(1)}}{U\rho_{\rm H}} \le -1,
\end{equation}
\begin{equation}
	\label{eq:condition 2}
	{\text{(Constraint 3)}}~~\theta_{\rm L}^{(U)} = \theta_{\rm SA}^{\prime} + \frac{2k_{\rm c}^{(U)}}{U\rho_{\rm L}} \ge 1.
\end{equation}
%Condition 1 and 2 ensure that the left edge of the first rainbow block covers -1, while the right edge of the last rainbow block covers 1.
As such, we can achieve seamless coverage in the angular domain $ [-1,1) $.
Next, a feasible TD parameter $ \theta_{\rm SA}^{\prime} $ satisfying the above constraints is given as follows.
\begin{lemma}[A feasible TD parameter $ \theta_{\rm SA}^{\prime} $]
	\label{lemma:solution}
	\emph {Given \eqref{eq:condition f_c}-\eqref{eq:condition 2}, a feasible TD parameter $ \theta_{\rm SA}^{\prime} $ can be set as}
	\begin{equation}
	\label{eq:solution}
		\theta_{\rm SA}^{\prime} = -1 + \frac{1-2\lceil k_{\rm th} \rceil}{U}, 
	\end{equation}
	\emph{where $ k_{\rm th} $ is represented as $ k_{\rm th} =  \frac{f_{\rm H}}{B} $.} 
\end{lemma}
\begin{proof}
	Based on \eqref{eq:condition f_c}, by substituting $ \theta_{\rm SA}^{\prime} = -1 + \frac{1-2k_{\rm c}^{(1)}}{U} $ into \eqref{eq:condition 1} and \eqref{eq:condition 2}, we have
	\begin{equation}
	\label{eq:simplified 1}
		k_{\rm c}^{(1)} \ge \frac{\rho_{\rm H}}{2\rho_{\rm H} - 2} = \frac{f_{\rm H}}{B},~~ k_{\rm c}^{(1)} \ge \frac{f_{\rm c} - \frac{B}{2}(2U-1)}{B}. 
	\end{equation}
	Then, \eqref{eq:condition 3} can be simplified as
	\begin{equation}
	\label{eq:simplified 2}
		k_{\rm c}^{(1)} \ge \frac{\rho_{\rm L}}{\rho_{\rm H} - \rho_{\rm L}} = \frac{f_{\rm L}}{B}.
	\end{equation}
	Combining \eqref{eq:simplified 1} and \eqref{eq:simplified 2}, we have 
	\begin{equation}
	\label{eq:k_th condition}
		k_{\rm c}^{(1)} \ge k_{\rm th} \triangleq \max\Big\{ \frac{f_{\rm L}}{B},\frac{f_{\rm c} - \frac{B}{2}(2U-1)}{B}, \frac{f_{\rm H}}{B} \Big\} \!=\! \frac{f_{\rm H}}{B}.
	\end{equation}
	
	By substituting \eqref{eq:k_th condition} into \eqref{eq:condition f_c}, we can obtain a feasible solution for the TD parameter $ \theta_{\rm SA}^{\prime} $, which is $ \theta_{\rm SA}^{\prime} = -1 + \frac{1-2\lceil k_{\rm th} \rceil}{U} $ and thus completing the proof.
\end{proof}

Based on Lemma \ref{lemma:solution}, we can achieve seamless angular coverage by setting $ \theta_{\rm SA}^{\prime} = -1 + \frac{1-2\lceil k_{\rm th} \rceil}{U} $.
However, the coverage region of each adjacent rainbow block may slightly overlap, leading to denser coverage at the boundaries\footnote{Indeed, the TD parameter value that we choose in \eqref{eq:solution} is the solution with the smallest overlapping coverage regions between rainbow blocks. It can be observed that the inter-rainbow interval $ {\rm RG}^{(u_{\rm c})} $ increases with $ k_c^{(1)} $. Hence, we choose $ k_{\rm c}^{(1)} = k_{\rm th} $ (see \eqref{eq:k_th condition}), resulting in small overlapping regions ($ |{\rm RG}^{(u_{\rm c})}| $).}.
Some of the beam angles in the overlapping coverage regions may be steered towards the same angle, w.r.t. $ \theta^{(u_{\rm c})}_{m_{\rm L}} = \theta^{(u_{\rm c}+1)}_{m_{\rm H}} $, where $ f_{m_{\rm L}} < f_{\rm c} $ and $ f_{m_{\rm H}} > f_{\rm c} $.
In this case, we have
\begin{equation}
	\label{eq:suppose equation}
	{k_{\rm c}^{(u_{\rm c})}} = \frac{f_{m_{\rm L}}}{f_{m_{\rm H}} - f_{m_{\rm L}}}.
\end{equation}

\begin{remark}[Overlapping coverage regions]
	\emph{In general, $ \frac{f_{m_{\rm L}}}{f_{m_{\rm H}}- f_{m_{\rm L}}} $ is not an integer, thus \eqref{eq:suppose equation} does not hold.
	However, the angular coverage of adjacent beams within the overlapping regions is highly close, resulting in inefficient utilization of splitted beam resources.
	Hence, too large overlapping regions between rainbow blocks need to be avoided.
	The maximum overlapping region is $ |RG^{(U-1)}| $, given by
		\begin{equation}
			\begin{aligned}
				|{\rm RG}^{(U-1)}| &= \frac{2f_{\rm c}}{Uf_{\rm L}f_{\rm H}}|-Bk_{\rm c}^{(U-1)} + f_{\rm L}| \\
				&{=}\frac{2f_{\rm c}}{Uf_{\rm L}f_{\rm H}}\Big|f_{\rm L} - B\Big(\lceil \frac{f_{\rm H}}{B} \rceil+ U-1\Big) \Big| \overset{(a)}{\approx} \frac{2B}{f_{\rm c}},
			\end{aligned}
		\notag
		\end{equation}
	where the approximation $(a)$ is due to $B\lceil \frac{f_{\rm H}}{B} \rceil\approx f_{\rm H}$ and $\frac{f_{\rm c}}{f_{\rm L}f_{\rm H}}\approx\frac{1}{f_{\rm c}}$. It can be observed that the maximum overlapping region is proportional to the relative bandwidth $ \frac{B}{f_{\rm c}} $.
	For the BS parameters in Example 1, we have $ |{\rm RG}^{(U-1)}| = 0.1 $, which is significantly smaller than the width of a rainbow block.
	Since the overlapping regions are relatively small, we can roughly assume that the $ MU $ beams uniformly cover the angular domain ranging from $ -1 $ to $ 1 $.}		
\end{remark}

We plot the multi-beam pattern of all subcarriers inside the rainbow blocks and around boundaries in Figs. \ref{fig:coverage}(a) and \ref{fig:coverage}(b), respectively.
The BS parameters are the same as that in Example 1.
Fig. \ref{fig:coverage}(a) shows that the beams inside the rainbow block uniformly cover specific angular regions and approximately $ U = 8 $ beams are steered towards an angular range of $ 0.02 $.
In contrast, the method in \cite{cui2022near} steers only one beam in the same space, demonstrating that the proposed SA-based method can significantly enhance the spatial resolution with the same spectrum resources.
Moreover, it can be observed from Fig. \ref{fig:coverage}(b) that although multiple beams formed at different subcarriers densely cover the boundaries of rainbow blocks, the angles of many beams are very close to each other, which fails to improve resolution, even in high signal-to-noise ratio conditions, leading to inefficient utilization of multi-beam resources.
\vspace{-5pt}
\section{Proposed Wideband Near-field Beam \\ Training Scheme}
\label{sec:proposed method}
In this section, we propose an effective wideband beam training method.
The proposed wideband beam training method is composed of three-stages namely, the \emph{angle sweeping using multiple rainbow blocks}, \emph{ambiguity elimination} and the \emph{range sweeping} as illustrated in Fig. \ref{fig:Algorithm}.
The key idea is to leverage the beam-split effect in both frequency and spatial domains of an activated S-ULA to extend the beam coverage regions and to achieve super-resolution angle estimation.
Then, the property, that the single beam at each subcarrier is approximately uniformly spaced in the angular domain, is utilized to select specific subcarriers with appropriate frequency intervals.
We leverage the single beam of these selected subcarriers to cover the candidate user angles and resolve the angular ambiguity by comparing the received power over the selected subcarriers.
Finally, by activating the entire XL-array, we control the splitted beams at all subcarriers to be focused at the estimated user angle but different ranges, which only needs one pilot to achieve range sweeping.
\vspace{-10pt}
\subsection{Angle Estimation}
The proposed angle estimation consists of two stages, including the angle sweeping and ambiguity elimination. 
In particular, the angular ambiguity arises from the beam-split effect in the spatial domain due to the sparsity of the activated S-ULA, which can be resolved by an activated central subarray.
\subsubsection{Angle Sweeping}
In Section \ref{Sec:Multi-beam Characteristic}, it is revealed that an activated central S-ULA possesses a beam-split effect in the spatial domain, which can be used for angle sweeping with super-resolution.
Motivated by this phenomenon, we first activate a central S-ULA comprising $ \widetilde{Q} $ antennas with $ Ud_0 $ inter-antenna spacing.
The TD beamformer is given by \eqref{eq:SA TD beamformer}, while the TD parameter $ \theta_{\rm SA}^{\prime} $ can be obtained in \eqref{eq:solution}.
Then, the received signal at the $ m $-th subcarrier is given by
\begin{equation}
	y_{m}^{\rm SA} = \sqrt{P_t} (\mathbf{h}_m^{\rm SA})^H \mathbf{w}_m\left(\theta_{\rm SA}^{\prime}, U \right) x_m+n_m.
\end{equation}
%where $ P_t $ and $ x_m $ denotes the transmit power and the pilot symbol with unit power, respectively.
%Moreover, $ n_m \sim \mathcal{C N}(0, \sigma^2) $ represents the additive white Gaussian noise (AWGN) with $ \sigma^2 $ denoting the noise power.
However, due to the increased path loss at higher frequencies, we define a \emph{calibrated} received signal power at the $m$-th subcarrier, given by~\cite{cui2022near}
\begin{equation}
	\label{eq:calibrated received power}
	P_m^{\rm SA} = |f_m y_{m}^{\rm SA}|^2
\end{equation} 
As such, the subcarrier with highest calibrated received power can be estimated by $ f_{\hat{m}}= \arg \max_{f_m} P_m^{\rm SA} $.
Further, due to the beam-split effect in the spatial domain, the $ k $-th candidate user angle is given by
\begin{equation}
\label{eq:candidate angles}
	\theta_{\hat{m}}^{(k)}  = \theta^{\prime}_{\rm SA} + \frac{2(k_{\hat{m}}^{(1)}+k-1)}{U\rho_{\hat{m}}}, 
\end{equation}
where $ k = 1,2,\cdots, |\mathcal{K}_{\hat{m}}| $ and $ k_{\hat{m}}^{(1)} = \lceil\frac{f_{\rm H}}{B}\rceil $.
Next, we present the method to determine the actual user angle from the $ |\mathcal{K}_{\hat{m}}| $ candidate angles in \eqref{eq:candidate angles}.
\begin{figure}
	\centering
	\includegraphics[width = 1\columnwidth]{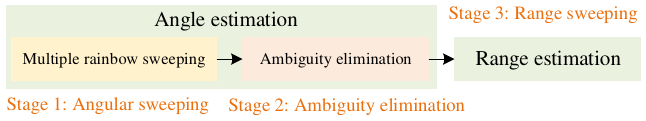}
	\caption{The framework of proposed wideband beam training method.}
	\label{fig:Algorithm}
	\vspace{-10pt}
\end{figure}
\subsubsection{Ambiguity Elimination}
\label{sec:Ambiguity Elimination}
Similar to the sparse-activation method, we activate a central subarray (dense array) consisting of $ Q $ antennas to resolve the angular ambiguity, for which the user is assumed to be located in the far-field region of the activated central subarray\footnote{As discussed earlier, the number of antennas in the activated central subarray should satisfy $ r_0 > R_{\rm Eff}^{\rm CS} = 0.367\frac{2(Q-1)^2d_{\rm c}^2}{\lambda_{\rm c}} $, and this condition is easily achievable.}.
Specifically, we select $|\mathcal{K}{\hat{m}}|$ subcarriers and use the corresponding single beams ($|\mathcal{K}{\hat{m}}|$ beams in total) formed at the selected subcarriers to cover the candidate user angles.
Then, according to the selected subcarriers with the highest calibrated received signal power, we can obtain the actual user angle.

\underline{\textbf{Channel model of the central subarray:}}
The LoS channel between the central subarray and the user at the $ m $-th subcarrier can be modeled as
\begin{equation}\label{Eq:ff-model}
	(\mathbf{h}^{\rm CS}_{m})^H = \sqrt{Q}\beta_m \mathbf{a}_m^{H}(\theta_{0}),
\end{equation}
where $ \mathbf{a}_m(\theta_{0}) $ denotes the far-field channel response vector of the activated central subarray, given by
\begin{equation}\label{far_steering}
	\left[\mathbf{a}_m^H(\theta_{0}) \right]_{q} = \frac{1}{\sqrt{Q}}e^{\jmath\frac{ 2 \pi }{\lambda_m}qd_{\rm c}\theta_{0}}, \forall q\in \mathcal{Q}.
\end{equation}
Herein, $ \mathcal{Q} \triangleq \{-\frac{Q-1}{2}, -\frac{Q-1}{2}+1, \cdots, \frac{Q-1}{2}\} $ denotes the antenna index set of the central subarray.
Moreover, the TD beamformer is given by
$$ [\mathbf{w}_m^H(\theta^{\prime}_{\rm CS})]_q = \frac{1}{\sqrt{Q}}e^{-{\jmath \frac{2 \pi}{\lambda_m} qd_{\rm c}\theta^{\prime}_{{{\rm CS}}}}}, \forall q\in \mathcal{Q}, $$
where $ \theta_{\rm CS}^{\prime} = \frac{{{q}}d_{\rm c}\tau_{{q}}}{c} $ denotes the adjustable TD parameter and $ \tau_{{q}} = \frac{{{q}}d_{\rm c}\theta_{\rm CS}^{\prime}}{c} $ represents the actual time-delay of the $ {q} $-th TD circuit in the activated central subarray.

Next, $ |\mathcal{K}_{\hat{m}}| $ single beams formed at the selected $ |\mathcal{K}_{\hat{m}}| $ subcarriers are required to be sequentially steered towards the candidate user angles.
To achieve this goal, we first analyze the array gain for the central subarray and the present the subcarrier selection mechanism. 

\underline{\textbf{Array gain for the central subarray:}}
Similar to \eqref{eq:SA array gain}, the array gain at the $ m $-th subcarrier is defined by
\begin{equation}
	\label{eq:CS array gain}
	\begin{aligned}
		f_m(\theta, \theta^{\prime}_{\rm CS}) =\left|\mathbf{w}_m^H(\theta^{\prime}_{\rm CS}) \mathbf{a}_m(\theta)\right|.
	\end{aligned}	
\end{equation}
Then, the angle of the beam formed at the $ m $-th subcarrier can be obtained as follows.
\begin{lemma}
	\label{lemma:beam angle for the central subarray}
	\emph{Given an activated central subarray with $ Q $ antennas and a TD beamformer $ \mathbf{w}_m(\theta^{\prime}_{\rm SA}) $, the steered beam angle $ \theta_{m} $ at the $ m $-th subcarrier is given by}
	\begin{equation}
		\label{eq:single beam}
		\theta_m = \theta^{\prime}_{\rm CS} + {2p_m}\frac{f_{\rm c}}{f_m},
	\end{equation}
	\emph {where $ p_{m} \in \Big\{ \mathbb{Z} ~ \cap \big\{k| \theta^{\prime}_{\rm CS} + {2k}\frac{f_{\rm c}}{f_m} \in [-1,1) \big\} \Big\} $}.
\end{lemma}
\begin{proof}
	The proof is similar to that of Lemma \ref{lemma:multi-beam} and thus omitted for brevity.
\end{proof}
\begin{figure*}[t]
	\centering
	\vspace{-10pt}
	\subfloat[Original subcarrier selection.]{
		\includegraphics[width=0.45\textwidth]{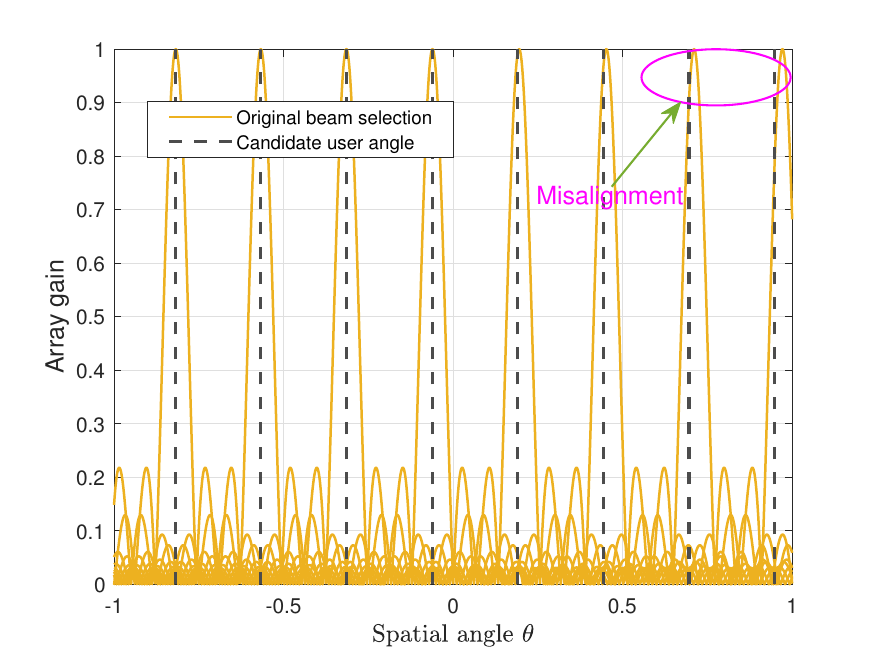}}
	\subfloat[Calibrated subcarrier selection.]{
		\includegraphics[width=0.45\textwidth]{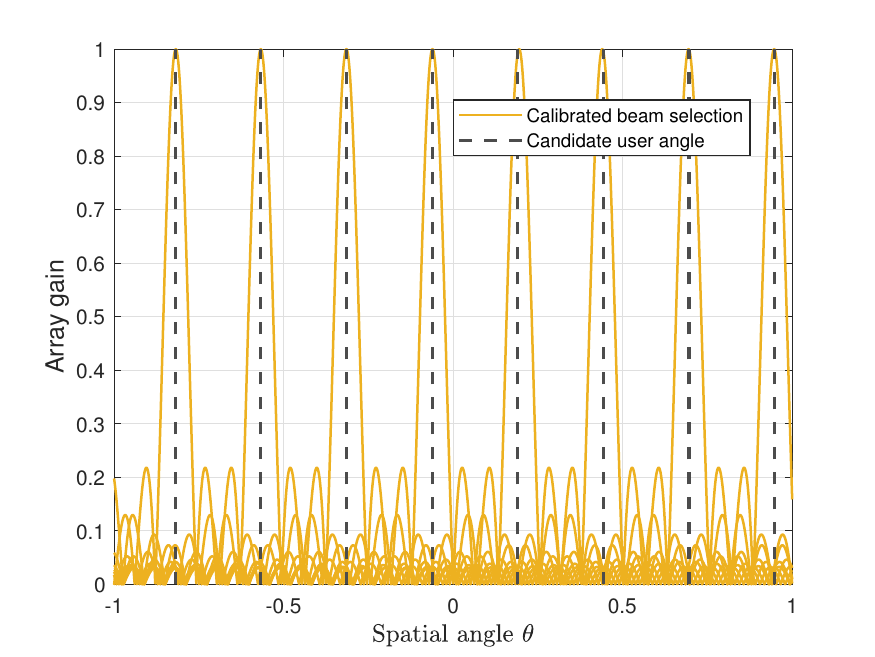}}
	\caption{Beam gains of the selected subcarriers, where $ f_{\rm c} = 60 $ GHz and $ B = 3 $ GHz, where the number of antennas of the central subarray is $ Q = 33 $.
	According to \eqref{eq:TD parameter for CS}, the TD parameter is set as $ \theta^{\prime}_{\rm CS} = -80.8199 $ with $ p = 40 $.}
	\label{fig:modify}
\end{figure*}

According to Lemma \ref{lemma:beam angle for the central subarray}, the period of the beams at the subcarriers $ f_{m_{\rm L}} \le f_{\rm c} $ is  $ \frac{2f_{\rm c}}{f_{m_{\rm L}}} > 2 $, which results in a single beam characteristic. 
In contrast, for the subcarriers $ f_{m_{\rm H}} \ge f_{\rm c} $, the beam period is $ \frac{2f_{\rm c}}{f_{m_{\rm H}}} < 2 $, potentially leading to multiple beams and introducing new angular ambiguities. 
Consequently, we only utilize the subcarriers $ f_{m_{\rm L}} \le f_{\rm c} $ to resolve the angular ambiguity introduced in the first stage.
Moreover, without loss of generality, we can set a proper TD parameter $ \theta^{\prime}_{\rm CS} $ to keep the same $ p_{m} $ for all subcarriers and the condition is given by
$$ -1 \le  \theta^{\prime}_{\rm CS} + {2p}\frac{f_{\rm c}}{f_{m_{\rm L}}} < 1, $$
where $ p_{m_{\rm L}} = p , \forall f_{m_{\rm L}} < f_{\rm c} $.
Then, the steered beam angle at the frequency $ f_{m_{\rm L}} $ can be modified as $ \theta_{m_{\rm L}} = \theta^{\prime}_{\rm CS} + {2p}\frac{f_{\rm c}}{f_{m_{\rm L}}} $.
Next, we introduce the subcarrier selection mechanism and the corresponding TD parameter value in detail.

\underline{\textbf{Subcarrier selection:}}
We first show that the single beam formed at each subcarrier is approximately uniformly spaced in the angular domain.
Then, this uniform property can be leveraged to select $ |\mathcal{K}_{\hat{m}}| $ subcarriers with appropriate frequency intervals, using their single beams to cover the candidate user angles with periodicity, thereby resolving the angular ambiguity.
\begin{lemma}[Beam angle difference]
	\label{lemma:beam angle difference}
	\emph{Given the central subcarrier frequency $f_{\rm c}$ and bandwidth $ B $ with $ M $ subcarriers, the angle difference between the single beams formed at adjacent subcarriers is given by}
	\begin{equation}
		\label{eq:beam gap}
		\Delta \theta  \approx  \frac{2p}{f_{\rm c}}\Delta f = \frac{2p}{f_{\rm c}} \frac{B}{M}.
	\end{equation}
\end{lemma} 
\begin{proof}
	We can regard the single-beam angle $ \theta_{m_{\rm L}} $ in \eqref{eq:single beam} as a function of frequency $f$, given by
	$$ \theta(f) = \theta^{\prime}_{\rm CS} + {2p}\frac{f_{\rm c}}{f}. $$
	Due to $ B \ll f_{\rm c} $, $ \theta(f) $ can be approximated by the first-order Taylor expansion at the central frequency $ f_{\rm c} $, given by
	\begin{equation}
	\label{eq:Taylor approximation}
		\theta(f) \approx \theta^{\prime}_{\rm CS} + {2p} - \frac{2p}{f_{\rm c}}(f - f_{\rm c}).
	\end{equation}
	Hence, the beam angle $ \theta_{m_{\rm L}} $ formed at the $ m_{\rm L} $-th subcarrier can be approximated as
	\begin{equation}
		\label{eq:period of frequency}
		\theta_{m_{\rm L}} \approx \theta^{\prime}_{\rm CS} + {4p} - \frac{2p}{f_c}f_{m_{\rm L}}.
	\end{equation}
	Then, the beam angle difference formed at adjacent subcarriers is $ \Delta \theta  \approx \frac{2p}{f_{\rm c}}\Delta f = \frac{2p}{f_{\rm c}} \frac{B}{M} $ and thus completing the proof.
\end{proof}

From Lemma \ref{lemma:beam angle difference}, we can select $ |\mathcal{K}_{\hat{m}}| $ subcarriers with uniform frequency difference of $ \eta = \lfloor\frac{M-1}{2|\mathcal{K}_{\hat{m}}|}\rfloor $ ranging from $ f_{\rm L} $ to $ f_{\rm c} $, resulting in a beam period of $ \eta \Delta \theta  $.
In particular, the selected subcarriers are given by
\begin{equation}
	\label{eq:selected frequency}
	f^{(k)} = f_{\rm c} - (k-1)\eta\frac{B}{M}, k = 1,2,\cdots, |\mathcal{K}_{\hat{m}}|.
\end{equation}  
To ensure that the single beams formed at the selected subcarriers accurately align with the candidate user angles, the TD parameter $ \theta^{\prime}_{\rm CS} $ should be set as
\begin{equation}
	\label{eq:sigle beam angle of the central subcarrier}
	\theta_{\rm c} = \theta^{\prime}_{\rm CS} + {2p} = \theta_{\hat{m}}^{(1)},
\end{equation}
\begin{equation}
\label{eq:equal period}
	\eta \Delta \theta = \frac{M-1}{2|\mathcal{K}_{\hat{m}}|} \frac{2p}{f_{\rm c}} \frac{B}{M} = \frac{2}{U\rho_{\hat{m}}}.
\end{equation}
\begin{lemma}[A feasible solution for the TD parameter $ \theta^{\prime}_{\rm CS} $]
	\emph {Given the constraints \eqref{eq:sigle beam angle of the central subcarrier} and \eqref{eq:equal period}, a feasible solution for the TD parameter $ \theta^{\prime}_{\rm CS} $ is given by}
	\begin{equation}
		\label{eq:TD parameter for CS}
		\theta_{\rm CS}^{\prime} = \theta_{\hat{m}}^{(1)} - 2\left \lfloor {\frac{2f_{\rm c}^2}{Bf_{\hat{m}}}} + 0.5 \right \rfloor.
	\end{equation}
\end{lemma}
\begin{proof}
	By substituting $ 2p = \theta_{\hat{m}}^{(1)} - \theta^{\prime}_{\rm CS} $ into \eqref{eq:equal period}, we have 
	$$ \theta_{\rm CS}^{\prime} = \theta_{\hat{m}}^{(1)} - \frac{4f_{\rm c}^2}{Bf_{\hat{m}}}, $$
	where $ p = {\frac{2f_{\rm c}^2}{Bf_{\hat{m}}}}$.
	Considering that $ p $ is an integer, a feasible solution for $ \theta_{\rm CS}^{\prime} $ can be approximated by $ \theta_{\rm CS}^{\prime} = \theta_{\hat{m}}^{(1)} - 2\left \lfloor {\frac{2f_{\rm c}^2}{Bf_{\hat{m}}}} + 0.5 \right \rfloor $ and thus we have completed the proof.
\end{proof}
\begin{remark}[Subcarrier selection calibration]
	\emph{Due to the first-order Taylor approximation in \eqref{eq:Taylor approximation}, some of the single beams formed at the selected subcarriers may slightly deviate from the candidate user angles.
	Given the approximate TD parameter $ \theta_{\rm CS}^{\prime} $ in \eqref{eq:TD parameter for CS}, we need to calibrate the subcarrier selection in \eqref{eq:selected frequency} and the calibrated subcarriers are given by}
	\begin{equation}
	\label{eq:modified subcarriers}
		f^{(k)} = \arg \min_{f_{m_{\rm L}}} \Big|\theta_{\hat{m}}^{(k)}-\big(\theta^{\prime}_{\rm CS} + {2p}\frac{f_{\rm c}}{f_{m_{\rm L}}}\big)\Big|,
	\end{equation}
	\emph{where $ p = \left \lfloor {\frac{2f_{\rm c}^2}{Bf_{\hat{m}}}} + 0.5 \right \rfloor $.}
\end{remark}
\begin{example}
	\emph{In Fig. \ref{fig:modify}, we plot the array gains of the selected subcarriers.
	The BS parameters are the same as that in Example 1.
	The optimal subcarrier obtained by the angle sweeping in the first stage is assumed to be $ f_{300} = 59.3774 $ GHz, corresponding to the candidate user angles $ -0.8199 + 0.2526k, k =1,2,\cdots,8 $.
	The subcarriers selected in \eqref{eq:selected frequency} are $ f^{(k)} = 30 - 0.1875k$ GHz.
	Fig. \ref{fig:modify}(a) shows that the beams formed at $ f^{(7)} $ and $ f^{(8)} $ without calibration deviate from the candidate user angles $ \theta_{\hat{m}}^{(7)} $ and $ \theta_{\hat{m}}^{(8)} $, which may degrade the wideband beam training performance.
	Fortunately, we can refine the selected frequencies via \eqref{eq:modified subcarriers}, for which $ f^{(6)} $, $ f^{(7)} $ and $ f^{(8)} $ are calibrated as $ 59.0698 $, $ 58.8853 $ and $ 58.7036 $ GHZ, respectively.
	It is observed from Fig. \ref{fig:modify}(b) that the single beams formed at the calibrated subcarriers precisely align with the candidate user angles, which improves the angle estimation accuracy.}
\end{example}

\underline{\textbf{User angle determination:}}
Given the TD parameter $ \theta_{\rm CS}^{\prime} $ in \eqref{eq:TD parameter for CS} and selected subcarriers in \eqref{eq:modified subcarriers}, the received signal at the $ k $-th selected subcarrier is given by
$$ y_{k}^{\rm CS} = \sqrt{P_t} (\mathbf{h}_k^{\rm SA})^H \mathbf{w}_k\left(\theta_{\rm CS}^{\prime}\right) x_k+n_k. $$
Similar to \eqref{eq:calibrated received power}, the selected subcarrier with highest calibrated power can be estimated as
\begin{equation}
	f^{(\hat{k})} = \arg \max_{f^{(k)}} |f^{(k)} y_{k}^{\rm CS}|^2,
\end{equation}
which corresponds to the estimated user angle, given by
\begin{center}
	\begin{framed}
		{\setlength\abovedisplayskip{0pt}
			\setlength\belowdisplayskip{0pt}
			\begin{equation}
				\theta^{\ast} = \theta^{\prime}_{\rm CS} + {2p}\frac{f_{\rm c}}{f^{(\hat{k})}}.
			\end{equation} 
	}\end{framed}
\end{center}

As such, in the third stage, we only need to scan the range domain in the estimated angle $ \theta^{\ast} $, which significantly reduces the near-field beam training overhead. 

\subsection{Range Estimation}
In this stage, we activate the entire XL-array antennas and control beams at all subcarriers to be focused on specific locations in the estimated user angle $ \theta^{\ast} $.
Specifically, we only need one pilot to achieve the range sweeping by exploiting the wideband beam-split effect, instead of the exhaustive search in the range domain~\cite{cui2022near}. 
\subsubsection{Array gain of the entire XL-array}
The array gain for the entire XL-array is given by
\begin{equation}
	\label{eq:XL-array array gain}
	\begin{aligned}
		f_m^{\rm AG} =\left|\mathbf{w}_m^H(\theta^{\prime},\mu^{\prime},\theta^{\prime}_p,\mu_p^{\prime}) \mathbf{b}_m(\theta, \mu)\right|.
	\end{aligned}	
\end{equation}
\begin{figure}
	\centering
	\vspace{-14pt}
	\includegraphics[width = 0.85\columnwidth]{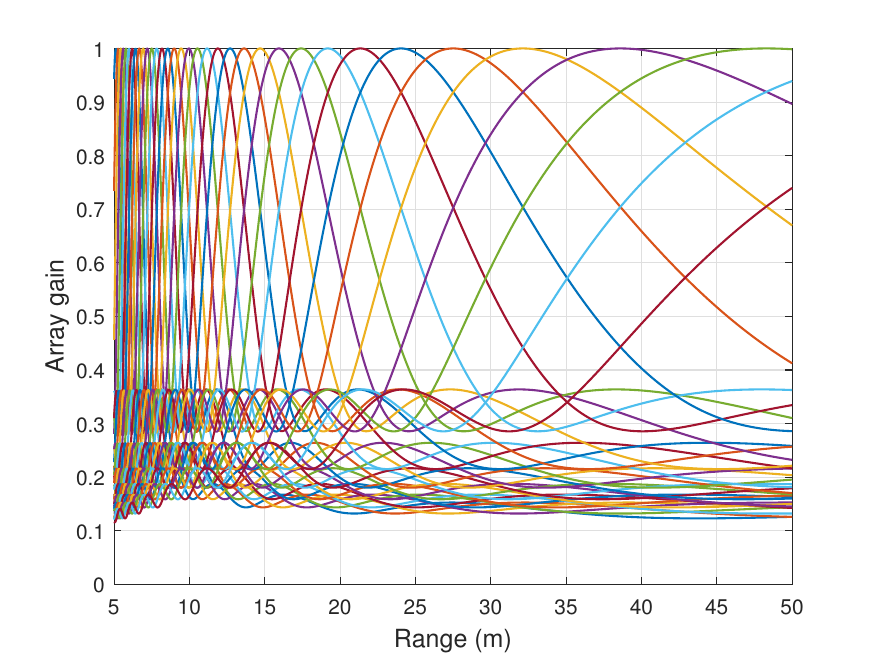}
	\caption{Beam coverage in the range domain, where $ f_{\rm c} = 60 $ GHz, $ B = 3 $ GHz and $ M = 1024 $. Moreover, the TD parameters are set as $ \theta^{\prime} = 0 $ and $ \mu^{\prime} = -1.8 $, while the PS parameters are $ \theta^{\prime}_p = 0 $ and $ \mu^{\prime}_p = 1.85 $. For clarity, we only present the subcarriers at frequency $ f_{1:30:1021} $.}
	\vspace{-10pt}
	\label{fig:RangeCoverage}
\end{figure}
Then, the beam focusing point can be obtained as follows.
\begin{lemma}[Beam focusing point]
	\label{lemma:beam focusing location}
	\emph{Given a TD-PS beamformer $ \mathbf{w}_m(\theta^{\prime},\mu^{\prime},\theta^{\prime}_p,\mu_p^{\prime}) $, the beam focusing point $ (\theta_{m}, \mu_m) $ at the frequency $ f_m $ can be obtained by}
	\begin{equation}
	\label{eq:focused angle}
		\theta_m = \theta^{\prime} + ({2q_m} + \theta_p^{\prime})\frac{f_{\rm c}}{f_m},
	\end{equation}
	\begin{equation}
	\label{eq:focused range}
		\mu_m = \mu^{\prime} + (\frac{{2s_m}}{d_{\rm c}}+\mu_p^{\prime})\frac{f_{\rm c}}{f_m},
	\end{equation}
	\emph {where $ q_{m} \in \big\{ \mathbb{Z} ~ \cap \big\{k| \theta^{\prime}_{\rm CS} + ({2k+\theta_p^{\prime}})\frac{f_{\rm c}}{f_m} \in [-1,1) \big\} \big\} $ and $ s_{m} \in \big\{ \mathbb{Z} ~ \cap \big\{k| \mu^{\prime} + {(2k+\mu_p^{\prime})}\frac{f_{\rm c}}{d_{\rm c}f_m} > 0 \big\} \big\} $. } 
\end{lemma}
\begin{proof}
	The array gain at the $ m $-th subcarrier in \eqref{eq:XL-array array gain} can be obtained as
	{\small
	\begin{align}
		f_m^{\rm AG} &=\frac{1}{N}\Big|\sum_{n \in \mathcal{N}} e^{\jmath n d_{\rm c}\big(\frac{2\pi}{\lambda_m}( \theta- \theta^{\prime})-\frac{2\pi}{\lambda_{\rm c}}\theta_p^{\prime} \big)-\jmath n^2 d_{\rm c}^2\big(\frac{2\pi}{\lambda_m}\left(\mu-\mu^{\prime}\right)-\frac{2\pi}{\lambda_{\rm c}}\mu_p^{\prime}\big)}\Big| \notag \\
		&= F \big(\frac{2\pi}{\lambda_m}(\theta -\theta^{\prime})-\frac{2\pi}{\lambda_{\rm c}}\theta_p^{\prime}, \frac{2\pi}{\lambda_m}(\mu -\mu^{\prime})-\frac{2\pi}{\lambda_{\rm c}}\mu_p^{\prime}\big),
	\end{align}}where $ F(x,y) \triangleq \frac{1}{N}\Big|\sum\limits_{n \in \mathcal{N}} e^{j n d_{\rm c} x-j n^2 d_{\rm c}^2 y}\Big| $. Moreover, it can be observed that $ F(x,y) $ is a periodic function with a period $ (\frac{2 \pi}{d_{\rm c}}, \frac{2 \pi}{d_{\rm c}^2}) $.
	Mathematically, we have $ F(x,y) = F(x-\frac{2q_m\pi}{d_{\rm c}},y-\frac{2 s\pi}{d_{\rm c}^2}) $ with $ q, s \in \mathbb{Z} $.	
	Given that $ f(\frac{2q\pi}{d_{\rm c}}, \frac{2s\pi}{d_{\rm c}^2}) = 1$ and $ \max F(x,y) = 1$, the beam focusing locations can be obtained by setting $ \frac{2\pi}{\lambda_m}(\theta -\theta^{\prime}) -\frac{2\pi}{\lambda_c}\theta_p^{\prime} = \frac{2q\pi}{d_{\rm c}} $ and $ \frac{2\pi}{\lambda_m}(\mu -\mu^{\prime})-\frac{2\pi}{\lambda_{\rm c}}\mu_p^{\prime} = \frac{2s\pi}{d_{\rm c}^2} $. As a consequence, the results are shown in \eqref{eq:focused angle} and \eqref{eq:focused range}. We thus complete the proof.
\end{proof}

Lemma \ref{lemma:beam focusing location} shows that when we set $ \theta^{\prime} = \theta^\ast $ and $ \theta^{\prime}_p = 0 $, the beams formed at all the subcarriers will be focused on the same angle $ \theta^\ast $ (similar in Section \ref{sec:normal TD parametr}), which is our desired TD-PS angle parameters.
Moreover, considering that $ \mu_m \ge \frac{\lambda_{\rm c}}{2D^2} $ and the beam period in the range domain is $ \frac{2f_{\rm c}}{f_md_{\rm c}} \gg \frac{\lambda_{\rm c}}{2D^2} $, only one beam will be formed at the $ m $-th subcarrier in the range domain.
Hence, without loss of generality, we can assume that all subcarriers share the same $ s_m = 0 $ for the TD parameter $ \mu^{\prime} $, satisfying $ \mu^{\prime} + {\mu_p^{\prime}}\frac{f_{\rm c}}{f_m} > 0 $.
%Then, we discuss the focused beam distribution in the range domain depending on $ \mu^{\prime} $ with another TD parameter fixed by $ \theta^{\prime} = \theta^\ast $.
%\begin{itemize}
%	\item {\bf Case 1}: $ \mu^{\prime} > 0 $, which is is termed as the normal ring-value. 
%	In this case, all subcarriers will be focused on the same point $ (\theta^{\prime}, \mu^{\prime}) $.
%	\item {\bf Case 2}: $ \mu^{\prime} \le 0$, which is termed as the abnormal ring-value. In this case, we can manipulate the $ M $ subcarriers distributed in the range domain, which may achieve higher resolution in the range sweeping, compared with~\cite{cui2022near}.
%\end{itemize}
\begin{remark}[Using only TD circuits]
	\emph{When the PSs are turned off, the beams of all subcarriers can not be guaranteed to focus within the desired range regime $[\mu_{\min},\mu_{\max}]$. In this case, the beam formed at the frequency $f_m$ is focused at $\mu_m = \mu^{\prime} + \frac{{2s_m}}{d_{\rm c}}\frac{f_{\rm c}}{f_m}$. By using the first-order Taylor expansion (similar to Lemma \ref{lemma:beam angle difference}), the focused range of the beam can be approximated as $\mu_m = \mu^{\prime}+\frac{4s_m}{d_c}-\frac{f_m}{f_c}\frac{2s_m}{d_c}$. As such, the beam range difference between the beams formed at adjacent subcarriers is $\Delta_{\mu} \approx \frac{4s_mB}{cM}$. When the bandwidth of the divided subcarriers (i.e., $\frac{B}{M}$) is large, $\Delta_{\mu}$ becomes too large, rendering only a small number of beams being focused within the desired rang regime $[\mu_{\min},\mu_{\max}]$. Considering a system setup where $B=3$ GHz and $M = 1024$ with users distributed within $[10,50]$ m, we obtain $|\Delta_{\mu}| \approx |\frac{4s_mB}{cM}|\ge \frac{4B}{cM} \approx 0.04$, while the desired range regime is $[\mu_{\min} = 0, \mu_{\min} = 0.1]$. Hence, only $\frac{\Delta_{\mu}}{\Delta_{\mu_0}} = \frac{0.1-0}{0.04} \approx 3$ beams can be effectively focused within the desired range regime, as illustrated in Fig.~\ref{fig:defocus}.}
\end{remark}
\begin{figure}
	\centering
	\vspace{-14pt}
	\includegraphics[width = 0.85\columnwidth]{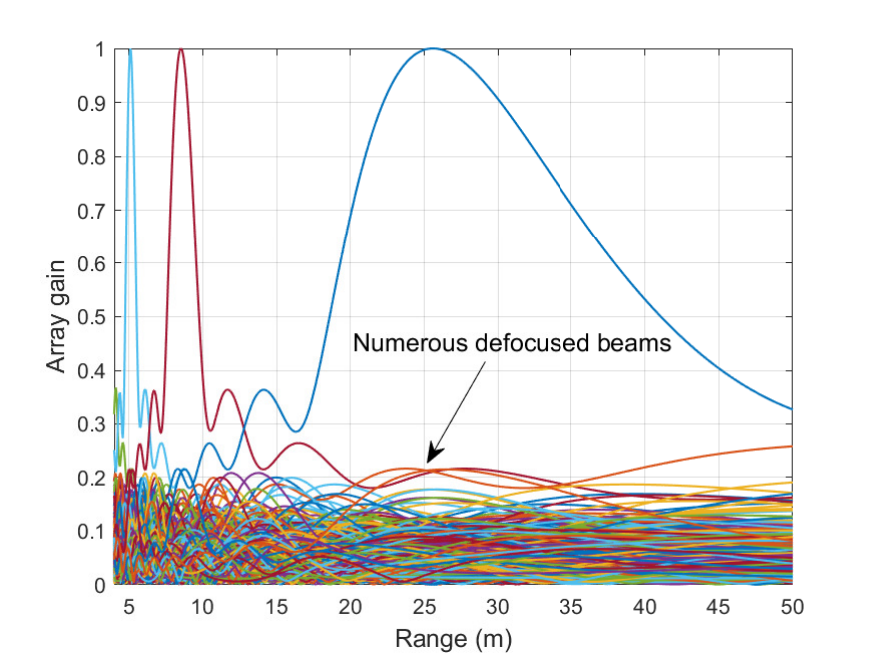}
	\caption{Beam coverage in the range domain with PSs turned off, where $ f_{\rm c} = 60 $ GHz, $ B = 3 $ GHz and $ M = 1024 $. The TD parameters are set as $ \theta^{\prime} = 0 $ and $ \mu^{\prime} = -\frac{2}{d_c} $.}
	\vspace{-10pt}
	\label{fig:defocus}
\end{figure}
\subsubsection{Subcarrier control in the range domain}
In fact, we require to control the subcarriers to be focused within a certain range regime $ [r_{\min}, r_{\max}] $, corresponding to the distance-ring range $ [\mu_{\min}, \mu_{\max}] $ with $ \mu_{\min} = \frac{\theta^{\ast}}{2r_{\max}} $ and $ \mu_{\max} = \frac{\theta^{\ast}}{2r_{\min}} $. 
To cover the range regime $ [r_{\min}, r_{\max}] $, we first fix the beam at the central subcarrier, which is focused on $ (\theta^{\ast}, \bar{\mu}) $, where $ \bar{\mu} = \frac{\mu_{\max}+\mu_{\min}}{2} $.
Mathematically, we have 
$$ \mu^{\prime} + \mu^{\prime}_p = \bar{\mu}. $$

As such, to cover the desired range regime, the following two conditions should be satisfied
\begin{equation}
\label{eq:left condition}
	\mu^{\prime} + \mu^{\prime}_p\frac{f_{\rm c}}{f_{\rm H}} \le \mu_{\min},
\end{equation}
\begin{equation}
\label{eq:right condition}
	\mu^{\prime} + \mu^{\prime}_p\frac{f_{\rm c}}{f_{\rm L}} \ge \mu_{\max}.
\end{equation}

A feasible solution that satisfies the above conditions for the TD-PS parameters $ \mu^{\prime} $ and $ \mu^{\prime}_p $ is given as follows.
\begin{lemma}[A feasible solution for TD parameter $ \mu^{\prime} $]
	\label{lemma:solution to mu}
	\emph {Given conditions \eqref{eq:left condition} and \eqref{eq:right condition}, a feasible solution for the TD-PS parameters $ \mu^{\prime} $ and $ \mu^{\prime}_p $ is given by}
	\begin{equation}
	\label{eq:solution to mu_p}
		\mu^{\prime}_p = \mu_{\rm th} \triangleq \max\Big\{\frac{\rho_{\rm L} (\mu_{\max}- \bar{\mu})}{1-\rho_{\rm L}}, \frac{\rho_{\rm H}(\bar{\mu} - \mu_{\min})}{\rho_{\rm H} - 1}\Big\}, 
	\end{equation}
	\begin{equation}
		\label{eq:solution to mu}
		\mu^{\prime} = \bar{\mu} - \mu^{\prime}_p. 
	\end{equation} 
\end{lemma}
\begin{proof}
	By substituting $ \mu^{\prime} = \bar{\mu} - \mu_p^{\prime} $ into \eqref{eq:left condition} and \eqref{eq:right condition}, we have
	\begin{equation}
	\label{eq:simplified mu 1}
		\mu_p^{\prime} \ge \frac{\rho_{\rm H} (\bar{\mu} - \mu_{\min})}{\rho_{\rm H} - 1},
	\end{equation}
	\begin{equation}
	\label{eq:simplified mu 2}
		\mu_p^{\prime} \ge \frac{\rho_{\rm L} (\mu_{\max}- \bar{\mu})}{1-\rho_{\rm L}}.
	\end{equation}
	Combining \eqref{eq:simplified mu 1} and \eqref{eq:simplified mu 2}, we have 
	\begin{equation}
	\label{eq:s_th condition}
		\mu_p^{\prime} \ge \mu_{\rm th} \triangleq \max\Big\{\frac{\rho_{\rm L} (\mu_{\max}- \bar{\mu})}{1-\rho_{\rm L}}, \frac{\rho_{\rm H}(\bar{\mu} - \mu_{\min})}{\rho_{\rm H} - 1}\Big\}.
	\end{equation}
	
	Hence, the PS parameter $ \mu^{\prime}_p $ can be set as $ \mu^{\prime}_p = \mu_{\rm th} $. Then, the TD parameter $ \mu^{\prime} $ is given by $ \mu^{\prime} =\bar{\mu} - \mu^{\prime}_p $ and thus completing the proof.
\end{proof}

We plot the beam pattern of a part of the subcarriers in the range domain in Fig. \ref{fig:RangeCoverage}.
The number of antennas of the entire XL-array is set by $ N = 513 $.
The other BS parameters are the same as that in Example 1.
For clarity, we only plot the beams formed at the subcarriers $ f_{1:30:1021} $.
The TD parameters are set by $ \theta^{\prime} = 0 $ and $ \mu^{\prime} = -1.7918 $, while the PS parameters are $ \theta^{\prime}_p = 0 $ and $ \mu^{\prime}_p = 1.8468 $.
In Fig. \ref{fig:RangeCoverage}, it can be observed that the beams formed at the above subcarriers are focused within the desired range $ [10, 50] $ m.
In fact, our average range estimation accuracy can reach $ \frac{M}{\Delta_r} $, where $ \Delta_r = r_{\max} - r_{\min} $.
However, since the subcarrier coverage density decreases as the distance increases (as shown in Fig. \ref{fig:RangeCoverage}), this resolution can not be achieved when users are uniformly distributed.
\subsubsection{User range determination}
Given the PS parameters ($ \theta_p^{\prime} = 0 $, $ \mu_p^{\prime} = \mu_{\rm th} $) and TD parameters ($ \theta^{\prime} = \theta^{\ast} $, $ \mu^{\prime} = \bar{\mu} - \mu_{\rm th} $), the received signal at the $ m $-th subcarrier can be represented by
$$ y_{m} = \sqrt{P_t} \mathbf{h}_m^H \mathbf{w}_m(\theta^{\ast}, \bar{\mu} - \mu_{\rm th}; 0, \mu_{\rm th}) x_m + n_m. $$

Similar to \eqref{eq:calibrated received power}, the subcarrier with the highest calibrated received power is given by
$$ f_{{m}^{\ast}} = \arg \max _{f_m} |f_m y_{m}|^2. $$
Herein, the corresponding estimated user range can be obtained as
\begin{center}
	\begin{framed}
		{\setlength\abovedisplayskip{0pt}
			\setlength\belowdisplayskip{0pt}
			\begin{equation}
				r^{\ast} = \frac{\theta^{\ast}}{2(\mu^{\prime}+\mu_{\rm th}\frac{f_{\rm c}}{f_{{m}^{\ast}}})}.
			\end{equation} 
	}\end{framed}
\end{center}
\subsection{Extensions and Discussions}
\underline{\textbf{Beam training overhead:}}
Since the mmWave and THz bands provide sufficient bandwidth, it is easy to achieve that the number of subcarriers exceeds the number of antennas.
Hence, only one pilot can be used to complete the beam sweeping in the first stage with a beam training overhead of $ T^{(1)} = 1 $.
Moreover, the beam training overhead in the second and third stages is given by $ T^{(2)} = 1 $ and $ T^{(3)} = 1 $.
As such, the overall beam training overhead is $ T = 3 $, which is ultra-low as compared with the benchmark schemes in Section \ref{Sec:benchmark schemes}, such as the exhaustive-search based beam training method and the near-field rainbow based wideband beam training scheme.
Indeed, in the first stage, we can trade the estimation accuracy for a reduction in the spectrum resources.
Specifically, by utilizing only $ \frac{1}{U} $ of the subcarriers, it is possible to achieve the same level of estimation accuracy as the non-sparse activation method such as {the near-field rainbow based beam training in~\cite{cui2022near}.
\begin{remark}[The impacts of the number of antennas]
	\emph{First, given the constraints of limited spectral resources, increasing the number of antennas in the XL-array reduces the range estimation accuracy of the proposed wideband beam training scheme. This degradation occurs because a larger array aperture substantially extends the Rayleigh distance, thereby enlarging the candidate coverage area. A feasible compensation strategy is to increase the system bandwidth (i.e., the number of subcarriers), thereby enhancing the beam coverage density within the user distribution region.
	Second, the number of antennas in the central subarray should be carefully designed. On one hand, an excessive number of antennas in the subarray can lead to the energy-spread effect, which degrades the accuracy of beam training. On the other hand, an insufficient number of antennas results in an overly broad beam-width, potentially covering two adjacent candidate user angles and introducing new angular ambiguities. Specifically, the number of antennas in the central subarray should satisfy the condition $U \le Q \le \sqrt{\frac{r_{\min}}{0.367d_c}}$, as detailed in \cite{zhou2024nearsparseDFT}.
	Third, a fundamental trade-off exists between the number of antennas in the activated S-ULA and the accuracy of angle estimation. Increasing the activation interval (i.e., reducing the number of activated antennas) leads to a higher number of rainbow blocks, thereby improving the angle estimation accuracy. However, this also reduces the array gain, making the proposed wideband beam training method more susceptible to noise and ultimately degrading angle estimation accuracy. A viable mitigation approach is to increase the transmit power of the BS.}
\end{remark}
\vspace{-10pt}
\begin{remark}[Estimation error] \label{remark:estimation_error}
	\emph{The estimation error in the proposed wideband beam training mainly arises from two effects.
	On one hand, the estimation accuracy of the proposed wideband beam training method is limited by the spectrum resources.
	On the other hand, as compared with the dense array, the proposed sparse-activation method sacrifices part of the array gain, making the proposed algorithm more sensitive to noise.
	Hence, increasing the transmit power and the number of subcarriers can effectively enhance the estimation accuracy of the proposed wideband beam training method.}
\end{remark}
\begin{remark}[Extension to the far-field scenario]
	\label{re:Far-field extension-scheme1}
	{\rm For the far-field scenario, we only need to perform the first two stages without the third stage. 
	Herein, the PSs are redundant since the TD beamforming architecture is sufficient for the beam control in the angular domain, which reduces the hardware cost and energy consumption in the far-field scenarios.
	Specifically, in the first stage, we sparsely activate the entire array instead of a central subarray, while the entire array can be activated in the second stage to eliminate the angular ambiguity, since the user is located in the far-field region of the entire array.
	As such, the TD parameters in the first and second stage can be similarly set as \eqref{eq:solution} and \eqref{eq:TD parameter for CS}, respectively.}
\end{remark}
\begin{remark}[Extension to multi-path channels]
	\label{re:multi-path channel}
	{\rm The proposed super-resolution wideband beam training method can also be applied to the general multi-path scenario.
	Specifically, in the first stage, the user can feed back $L$ indices of subcarriers corresponding to the $L$ highest received power signals, where $L$ represents the number of multi-path components.
	Then, we can successively resolve the angular ambiguity for each path in the second stage and this procedure requires $ L $ pilots. 
	Finally, the range sweeping is performed at the estimated angle for each of the $ L $ paths, which needs $ L $ beam training pilots.
	Overall, when our proposed wideband beam training method is applied to the multi-path channel, the beam training overhead becomes $ T = 1+2L $.
	Fortunately, in the mmWave or THz bands, the number of multi-paths is small, usually ranging from 2 to 4 \cite{liu2021thz}.
	As a result, this does not significantly increase the beam training overhead of the proposed method.} 
\end{remark}
\begin{figure}[!t]
	\centering
	\vspace{-14pt}
	\includegraphics[width = 0.85\columnwidth]{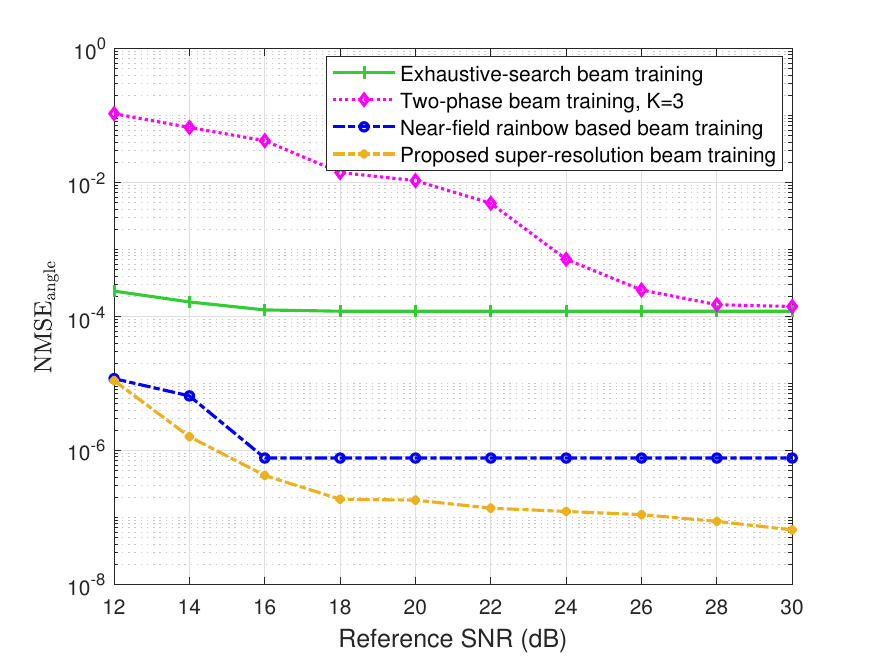}
	\caption{Angle estimation NMSE versus the reference SNR, where $ f_{\rm c} = 60 $ GHz, $ B = 3 $ GHz and $ M = 1024 $.}
	\label{fig:NMSE angle}
	\vspace{-10pt}
\end{figure}
\begin{figure}[!t]
	\centering
	\includegraphics[width = 0.85\columnwidth]{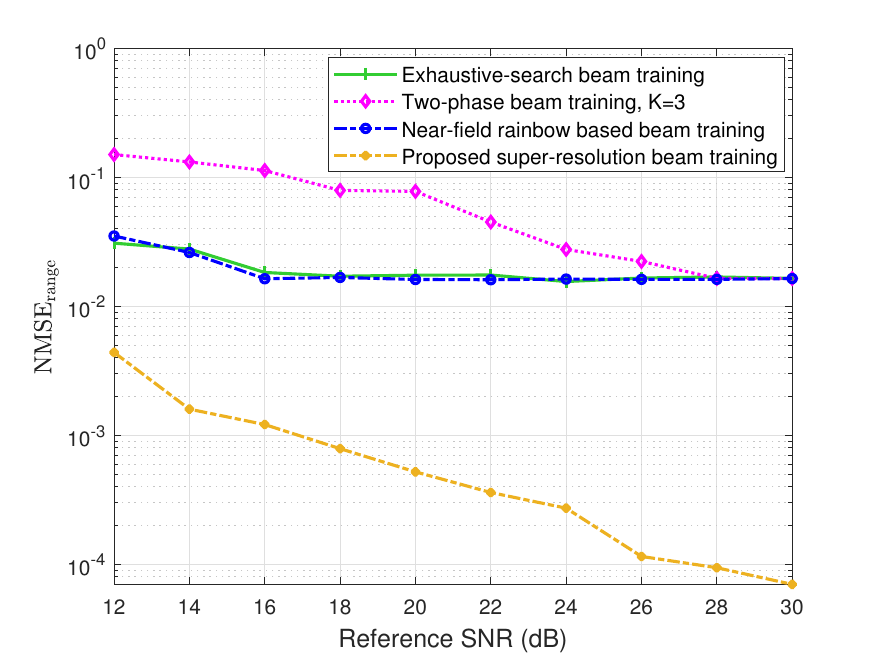}
	\caption{Range estimation NMSE versus the reference SNR, where $ f_{\rm c} = 60 $ GHz, $ B = 3 $ GHz and $ M = 1024 $.}
	\vspace{-10pt}
	\label{fig:NMSE range}
\end{figure}
\section{Numerical Results}
\label{Sec:numericalResults}
In this section, numerical results are provided to demonstrate the effectiveness of our proposed super-resolution wideband beam training method for near-field communications.
\subsection{System Setup and Benchmark Schemes}
\label{Sec:benchmark schemes}
The BS system parameters are set as follows.
We assume that the BS is equipped with $ N = 513 $ antennas and we activate a central S-ULA with the activation interval $ U = 8 $, while the number of antennas of the activated central dense subarray in the second stage is $ Q = 129 $.
The carrier frequency is $ f_c =  60$ GHz and the bandwidth is $ B = 3 $ GHz with $ M = 1024 $ subcarriers.
The transmit power of the BS and the noise power are set as $ P_t = 30 $ dBm and $ \sigma^2 = -80 $ dBm, respectively.
Moreover, the reference SNR at $m$-th subcarrier is defined by $ {\rm SNR}_m = \frac{\widetilde{Q}P_t\beta_m}{r_0^2\sigma^2} $, while the normalized mean square errors (NMSEs) for the angle and range estimation are defined as $ {\rm NMSE}_{\rm angle} = $  $\frac{\mathbb{E}\big( |\theta_0 - \theta^{\ast}|^2\big)}{\mathbb{E}\big( |\theta_0|^2 \big)}$ and $ {\rm NMSE}_{\rm range} = \frac{\mathbb{E}\big( |r_0 - r^{\ast}|^2\big)}{\mathbb{E}\big( |r_0|^2 \big)}$, respectively.
For performance comparison, we consider the following benchmark schemes.
\begin{figure}
	\centering
	\vspace{-14pt}
	\includegraphics[width = 0.85\columnwidth]{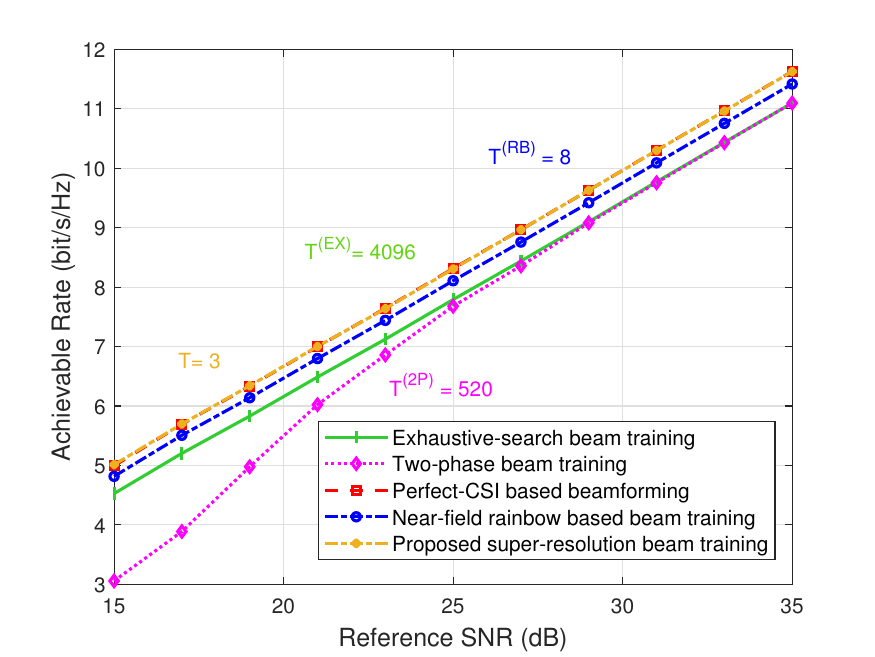}
	\caption{Achievable rate versus reference SNR, where $ f_c = 60 $ GHz, $ B = 3 $ GHz and $ M = 1024 $.}
	\vspace{-10pt}
	\label{fig:Rate versus SNR}
\end{figure}
\begin{figure}
	\centering
	\includegraphics[width = 0.85\columnwidth]{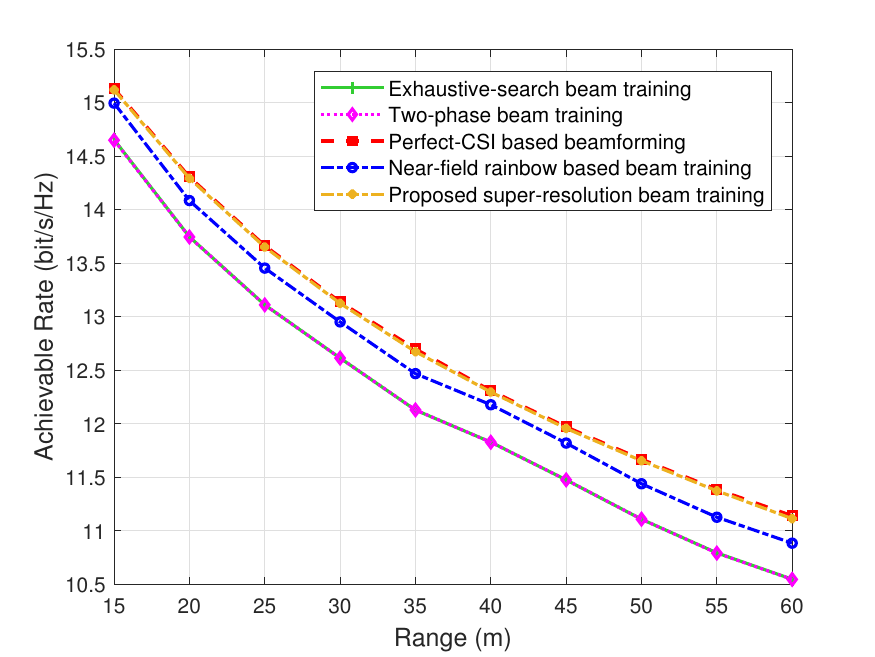}
	\caption{Achievable rate versus user range, where $ f_{\rm c} = 60 $ GHz, $ B = 3 $ GHz and $ M = 1024 $.}
	\vspace{-10pt}
	\label{fig:Rate versus range}
\end{figure}
\begin{itemize}
	\item {\rm \textbf{Perfect-CSI based beamforming:}}
	This method assumes that the BS knows the CSI in prior and the beamforming vector is $ {\mathbf{w}} =  \mathbf{b}(r_{0},\theta_{0})$, which perfectly aligns the near-field channel of the typical user. 
	As such, this scheme serves as the performance upper bound for the beam training methods.
	\item {\rm \textbf{Exhaustive-search based beam training:}}
	For the exhaustive-search beam training method, the polar-domain codebook proposed in \cite{Cui2022channel} is applied, which is given by $ {\mathcal{W}}_{\rm{Pol}} = \{{\mathcal{W}}_1,\cdots,{\mathcal{W}}_{\bar n},\cdots,{\mathcal{W}}_{N} \} $.
	Specifically, we have $ {\mathcal{W}}_{\bar n} = \{ {\mathbf{w}}_{\bar{n},1} \cdots,{\mathbf{w}}_{\bar{n},\bar v},\cdots {\mathbf{w}}_{\bar{n},V}\}$ with  $ {\mathbf{w}}_{{\bar n},\bar v}=  \mathbf{b}\left(r_{{\bar n}, \bar v}, \theta_{\bar n}\right)$, where $ \theta_{\bar{n}} = \frac{2 {\bar n}-N+1}{N}, \forall {\bar n} \in \bar{\mathcal{N}} \triangleq \{1,2, \cdots, N\} $ and $ r_{{\bar n},{\bar v}}=\frac{1}{{\bar v}} \alpha_{\Delta}\left(1-\theta_{{\bar n}}^2\right), \quad \forall \bar v\in{\bar {\mathcal{V}}}\triangleq\{1,2,3, \cdots V\} $ represent the sampling angles and ranges, respectively. Moreover, $ V $ and $ \alpha_{\Delta} $ denote the number of sampling ranges and a constant corresponding to the quantization loss in the range domain, respectively \cite{Cui2022channel}.
	As such, the BS sequentially transmits pilot signals to the user with codewords $ {\mathbf{w}}_{{\bar n},\bar v} $ in $ {\mathcal{W}}_{\rm{Pol}} $, which steers the beam at $ (\theta_{\bar{n}}, r_{{\bar n},{\bar v}}) $.
	Then, via comparing the received signal power, the user location can be estimated with a beam training overhead $ T^{\rm (EX)} = NV $.
	\item {\rm \textbf{Near-field rainbow based beam training:}}
	Compared with the exhaustive-search beam training method, the near-field rainbow based scheme exploits the beam-squint effect to control the beams formed at different subcarriers focused on multiple angles within the same range-ring \cite{cui2022near}. 
	In particular, by carefully designing the TD parameters, $ M $ beams will be focused on a specific range-ring $ \alpha_{\bar v} =  \frac{{\bar v}}{2\alpha_{\Delta}}, \forall \bar v\in{\bar {\mathcal{V}}}$.
	Hence, the overhead of the near-field rainbow based beam training method is the number of the range-rings, i.e, $T^{\rm{(RB)}} = V$, which is the number of sampling ranges in the polar-domain codebook $ {\mathcal{W}}_{\rm{Pol}} $.
	\item {\rm \textbf{Two-phase beam training:}}
	The two-phase beam training method proposed in \cite{two_phase} exploited the energy-spread effect to decouple the angle and range estimation.
	In particular, the angle is estimated by the middle of the $3$-dB angular support of the received beam pattern, which can be obtained by the beam sweeping with the well-known DFT codebook.
	Then, the range can be estimated by the beam sweeping with the polar-domain codebook in the estimated user angle. 
	Hence, the overhead of two-phase beam training method is given by $T^{\rm{(2P)}} = N + V$.
    However, the median angle of $3$-dB angular support is not accurate enough due to the power fluctuation and received noise \cite{two_phase}.
    Then, an effective middle-$K$ angle selection scheme which selects $K$ candidate angles around the middle of the the $3$-dB angular support (instead of selecting only one middle angle) was proposed in \cite{two_phase}.
    As such, range sweeping is required to perform in the $K$ candidate user angles and the beam training overhead of the two-phase is recast as $T^{\rm{(2P)}} = N + KV$.
\end{itemize}
\vspace{-5pt}
\subsection{Performance Analysis}
In Fig. \ref{fig:NMSE angle}, we plot the angle estimation NMSE versus the reference SNR.
Several important observations can be summarized as follows.
First, it is shown that the angle estimation NMSE of the proposed wideband beam training method decreases with the reference SNR and is significantly lower than that of all benchmark schemes.
This is because the proposed method achieves higher resolution via using the activated S-ULA to form much more beams in the angular domain and thus causes smaller angle estimation error.
Second, at the low-SNR regime, the angle estimation NMSE of the near-field rainbow based beam training scheme approaches that of the proposed method, which can be explained that our proposed wideband beam training method is generally more sensitive to the noise due to the sparse activation (sacrificing array gain), thereby resulting in lower estimation accuracy as expected.

Fig. \ref{fig:NMSE range} shows the range estimation NMSE versus the reference SNR.
It is observed that the range estimation NMSE of our proposed method significantly outperforms all benchmarks schemes especially at the high-SNR regime.
This can be explained that the benchmark schemes are on-grid beam training methods and utilize predefined codebooks with a limited size to perform the range estimation, which relies on the number of sampling ranges $ V $, thereby resulting in a lower estimation accuracy. 
However, the proposed method leverages a large number of beams formed at different subcarriers to be focused on the desired angle and the range estimation accuracy depends on the number of subcarriers.
Hence, the proposed wideband beam training scheme can achieve a better range estimation accuracy due to $ M \gg L $.

In Fig. \ref{fig:Rate versus SNR}, we plot the achievable-rate performance of different beam training schemes versus the reference SNR.
First, it is observed that the achievable-rate performance of the proposed wideband beam training method outperforms the benchmark schemes and approaches that of the perfect-CSI based beamforming.
This is because our proposed wideband beam training method achieves a super-resolution estimation in both angle and range domains.
Moreover, in the low-SNR regime, the achievable-rate performance of the two-phase beam training scheme is significantly worse than that of other methods.
This is because the two-phase beam training method uses the energy-spread effect to estimate the user angle and achieves lower angle estimation accuracy at the low-SNR regime.

In Fig. \ref{fig:Rate versus range}, we plot the achievable-rate performance versus the user range.
In particular, the typical user is uniformly distributed within the angular sector $ \theta \in [-0.5, 0.5] $.
Several interesting observations are made as follows.
First, the achievable-rate performance of all schemes decreases with the user range.
This is because the reference SNR decreases with the user range, which degrades the achievable-rate performance.
Second, we observe that, for all user ranges, the proposed wideband beam training method outperforms all near-field beam training schemes w.r.t. achievable-rate performance, due to higher angle and range estimation accuracy achieved through the proposed sparse-activation method.
\section{Conclusions}
\label{Sec:Conclusions}
In this paper, we proposed a three-stage near-field wideband beam training method. 
We first studied the phenomenon named rainbow blocks due to the beam-split effect of an activated S-ULA in both spatial and frequency domains, which can be used to significantly extend the beam coverage regions in the angular domain.
However, inevitable angular ambiguity was introduced by the beam-split effect in the spatial domain.
To resolve this ambiguity, we selected specific subcarriers to steer single beams towards candidate user angles in the second stage and the actual user angle was estimated by comparing the received signal power over the selected subcarriers.
Finally, in the third stage, by means of the TD-PS architecture, the one-pilot range estimation was achieved by controlling the splitted beams over all subcarriers focused on the estimated angle but different ranges.
Numerical results demonstrated the effectiveness of the proposed beam training method, which achieved more accurate angle and range estimation with only three pilots as compared to existing benchmark schemes.

\bibliographystyle{IEEEtran}
\bibliography{IEEEabrv}

\end{document}

%% file: header.tex
\newtheorem{definition}{\emph{\underline{Definition}}}

\newtheorem{lemma}{\emph{\underline{Lemma}}}

\newtheorem{example}{\bf Example}
\newtheorem{remark}{\bf \emph{\underline{Remark}}}

\def\({\left(}
\def\){\right)}

\def\b0{{\mathbf{0}}}

\newcommand{\nn}{\nonumber}